%% file: main.tex
\documentclass[11pt, letterpaper]{article}
\usepackage[utf8]{inputenc}
\usepackage{fullpage,times}

\usepackage[dvipsnames]{xcolor}

\usepackage[procnumbered,ruled,vlined,linesnumbered]{algorithm2e}

\DontPrintSemicolon
\SetKw{KwAnd}{and}
\SetProcNameSty{textsc}
\SetFuncSty{textsc}
\usepackage{mathbbol}

\usepackage{amsmath,amsfonts,amsthm,amssymb,multirow}
\usepackage{times}
\usepackage{wrapfig}
\usepackage{graphicx}
\usepackage{enumitem}
\usepackage{subcaption} 
\usepackage{calc}
\usepackage{tikz}
\usepackage{hyperref}
\usetikzlibrary{decorations.markings, patterns}
\tikzstyle{vertex}=[circle, draw, inner sep=0pt, minimum size=4pt, fill = black]

\usepackage{graphicx}
\usepackage{tabularx}
\makeatletter
\newcommand{\multiline}[1]{%
  \begin{tabularx}{\dimexpr\linewidth-\ALG@thistlm}[t]{@{}X@{}}
    #1
  \end{tabularx}
}
\makeatother
\usepackage{mathtools}
\mathtoolsset{showonlyrefs}

\makeatletter
\def\BState{\State\hskip-\ALG@thistlm}
\makeatother

\newcommand{\ceil}[1]{\lceil #1 \rceil}
\newcommand{\floor}[1]{\lfloor #1 \rfloor}

\usepackage[compact]{titlesec}
\titlespacing{\section}{0pt}{3ex}{2ex}
\titlespacing{\subsection}{0pt}{2ex}{1ex}
\titlespacing{\subsubsection}{0pt}{0.5ex}{0ex}

\newtheorem{theorem}{Theorem}[section]

\newtheorem{corollary}[theorem]{Corollary}

\newtheorem{definition}[theorem]{Definition}
\newtheorem{lemma}[theorem]{Lemma}
\newtheorem{claim}[theorem]{Claim}

\makeatletter
\let\c@fconjecture\c@conjecture
\makeatother

\makeatletter
\let\c@fconj\c@conj
\makeatother

\def \eps {\varepsilon}

\newcommand{\ignore}[1]{}

\def \popped { \text{\rm popped} }

\def\defeq{\coloneqq}

\title{Hardness of Approximate Diameter: Now for Undirected Graphs}
\author{Mina Dalirrooyfard\thanks{Department of Electrical Engineering and Computer Science and CSAIL, MIT. minad@mit.edu. Research supported by the National Science Foundation (NSF) under NSF CAREER Award 1651838} ,\, Ray Li\thanks{Department of Computer Science, Stanford University. rayyli@cs.stanford.edu. Research supported by the National Science Foundation (NSF) under Grant No. DGE - 1656518 and by Jacob Fox's Packard Fellowship.} ,\, Virginia Vassilevska Williams\thanks{Department of Electrical Engineering and Computer Science and CSAIL, MIT. virgi@mit.edu. Supported by NSF CAREER Award 1651838, NSF Grants CCF-1528078, CCF-1514339 and CCF-1909429, BSF Grant BSF:2012338, a Google Research Fellowship and a Sloan Research Fellowship.}}
\date{}
\begin{document}

\maketitle
\thispagestyle{empty}

\begin{abstract}
Approximating the graph diameter is a basic task of both theoretical and practical interest. A simple folklore algorithm can output a 2-approximation to the diameter in linear time by running BFS from an arbitrary vertex. It has been open whether a better approximation is possible in near-linear time. A series of papers on fine-grained complexity have led to strong hardness results for diameter in directed graphs, culminating in a recent tradeoff curve independently discovered by [Li, STOC'21] and [Dalirrooyfard and Wein, STOC'21], showing that under the Strong Exponential Time Hypothesis (SETH), for any integer $k\ge 2$ and $\delta>0$, a $2-\frac{1}{k}-\delta$ approximation for diameter in directed $m$-edge graphs requires $mn^{1+1/(k-1)-o(1)}$ time. In particular, the simple linear time $2$-approximation algorithm is optimal for directed graphs.

In this paper we prove that the same tradeoff lower bound curve is possible for undirected graphs as well, extending results of [Roditty and Vassilevska W., STOC'13], [Li'20] and [Bonnet, ICALP'21] who proved the first few cases of the curve, $k=2,3$ and $4$, respectively. Our result shows in particular that the simple linear time $2$-approximation algorithm is also optimal for undirected graphs. To obtain our result we develop new tools for fine-grained reductions that could be useful for proving SETH-based hardness for other problems in undirected graphs related to distance computation.
\end{abstract}
\newpage
\setcounter{page}{1}

\input{intro}

\input{prelim}
\input{k4}

\input{proofsketch}

\input{simplified}

\section*{Acknowledgments} \noindent We thank Nicole Wein for helpful discussions, and we thank anonymous reviewers for helpful comments that improved the writing. 

%%%%%%%%%%%%%%%%%%%%%%%%%%%%%

\bibliographystyle{alpha}
\bibliography{bib} 

\appendix
\input{k5}
\end{document}

%% file: intro.tex
\section{Introduction}
\label{sec:intro}
One of the most basic graph parameters, the {\em diameter} is the largest of the shortest paths distances between pairs of vertices in the graph. Estimating the graph diameter is important in many applications (see e.g. \cite{diam-prac2,diam-prac4,diam-prac5}. For instance, the diameter measures how fast information spreads in networks, which is central for paradigms such as distributed and sublinear algorithms.

The fastest known algorithms for computing the diameter of an $n$-node, $m$-edge graph with nonnegative edge weights solve All-Pairs Shortest Paths (APSP) and
 run in $O(\min\{mn+n^2\log\log n, n^3/exp(\sqrt{\log n})\})$ time \cite{Pettie04,ryanapsp}. For unweighted graphs one can use fast matrix multiplication \cite{vstoc12,legallmult,almanvw21,seidel1995all,alon1997exponent} and solve the problem in $O(n^{2.373})$ time.

Any algorithm that solves APSP naturally needs $n^2$ time, just to output the $n^2$ distances. Meanwhile, the diameter is a single number, and it is apriori unclear why one would need $n^2$ time, especially in sparse graphs, for which $m\leq n^{1+o(1)}$.

There is a linear time folklore algorithm that is guaranteed to return an estimate $\hat{D}$ for the diameter $D$ so that $D/2\leq \hat{D}\leq D$, a so called $2$-approximation. The algorithm picks an arbitrary vertex and runs BFS from it, returning the largest distance found. The same idea achieves a near-linear time $2$-approximation in directed and nonnegatively weighted graphs by replacing BFS with Dijkstra's algorithm to and from the vertex.

Roditty and Vassilevska W. \cite{rv13}, following Aingworth, Chekuri, Indyk and Motwani \cite{aingworth}, designed a $3/2$-approximation algorithm running in $\tilde{O}(m\sqrt n)$ time, for the case when the diameter is divisible by $3$, and with an additional small additive error if it is not divisible by $3$. Chechik, Larkin, Roditty, Schoenebeck, Tarjan and Vassilevska W. \cite{ChechikLRSTW14} gave a variant of the algorithm that runs in $\tilde{O}(m^{3/2})$ time and always achieves a $3/2$-approximation (with no additive error). These algorithms work for directed or undirected graphs with nonnegative edge weights.

Cairo, Grossi and Rizzi \cite{CGR} extended the techniques of \cite{rv13} and developed an approximation scheme that for every integer $k\geq 0$, achieves an ``almost'' $2-1/2^k$-approximation (i.e. it has an extra small additive error, similar to \cite{rv13})
and runs in $\tilde{O}(mn^{1/(k+1)})$ time. The scheme only works for undirected graphs, however.

These are all the known approximation algorithms for the diameter problem in arbitrary graphs: the scheme of \cite{CGR,rv13} for undirected graphs, and the three algorithms for directed graphs: the exact $\tilde{O}(mn)$ time algorithm using APSP, the $\tilde{O}(m)$ time $2$-approximation and the $3/2$-approximation algorithms of \cite{rv13,ChechikLRSTW14}.
In Figure \ref{fig:history} the known algorithms are represented as purple and pink points.

A sequence of works \cite{rv13,BRSVW18, Li20,Bonnet21,Li21,DW21,Bonnet21ic} provided lower bounds for diameter approximation, based on the Strong Exponential Time Hypothesis (SETH) \cite{ipz1,CIP10} that CNF-SAT on $n$ variables and $O(n)$ clauses requires $2^{n-o(n)}$ time. The first such lower bound by \cite{rv13} showed that any $3/2-\eps$ approximation to the diameter of a directed or undirected unweighted graph for $\eps>0$, running in $O(m^{2-\delta})$ time for $\delta>0$, would refute SETH, and hence the \cite{rv13} $3/2$-approximation algorithm has a (conditionally) optimal approximation ratio for a subquadratic time algorithm for diameter.
Later, Backurs, Roditty, Segal, Vassilevska W. and Wein \cite{BRSVW18} showed that under SETH, any $O(m^{3/2-\delta})$ time algorithm can at best achieve a $1.6$-approximation to the diameter of an undirected unweighted graph. Thus, the \cite{rv13} $3/2$-approximation algorithm has a (conditionally) optimal running time for a $(1.6-\eps)$-approximation algorithm.

Following work of Li \cite{Li20} and Bonnet \cite{Bonnet21}, Li \cite{Li21} and independently Dalirrooyfard and Wein \cite{DW21}, provided a scheme of tradeoff lower bounds for diameter in {\em directed} graphs. They showed that under SETH, for every integer $k\geq 2$, a $(2-1/k-\eps)$-approximation algorithm for $\eps>0$ for the diameter in $m$-edge directed graphs, requires at least $m^{1+1/(k-1)-o(1)}$ time.
Thus in particular, under SETH, the linear time $2$-approximation algorithm for diameter is optimal for directed graphs.

\begin{figure}
    \centering
    \includegraphics{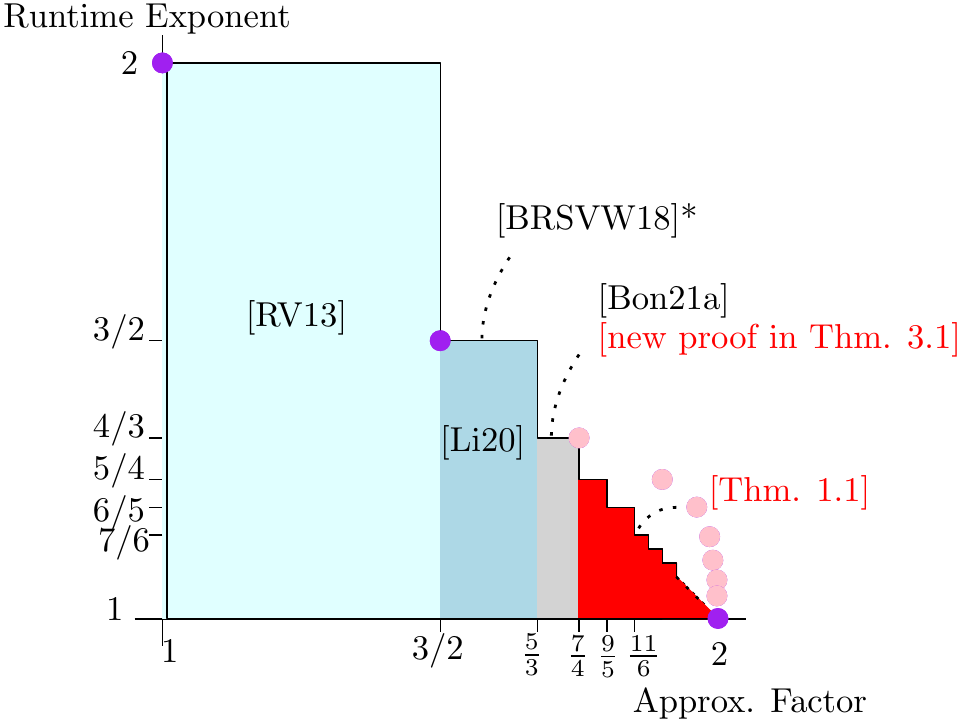}
    \caption{Hardness results for Diameter (in undirected graphs). The $x$-axis is the approximation factor and the $y$-axis is the runtime exponent. Black lines represent lower bounds. Purple dots represent algorithms, and pink dots represent algorithms that also lose an additive factor. The red region represents new approximation ratio vs. runtime tradeoffs that are unachievable (under SETH) in light of this work.
    In \cite{BRSVW18}, the labeled lower bound was proved for weighted graphs, and in unweighted graphs they proved a weaker lower bound that a $1.6-\varepsilon$ approximation needs $m^{3/2-o(1)}$ time.}
    \label{fig:history}
\end{figure}

For undirected graphs, however, only three conditional lower bounds are known: the $m^{2-o(1)}$ \cite{rv13}  lower bound for $(3/2-\eps)$-approximation, the $m^{3/2-o(1)}$ \cite{Li20}  lower bound for $(5/3-\eps)$-approximation, and the $m^{4/3-o(1)}$ \cite{Bonnet21ic}  lower bound for $(7/4-\eps)$-approximation (see Figure~\ref{fig:history}).

The tradeoff lower bounds for directed diameter of \cite{DW21} and \cite{Li21} crucially exploited the directions of the edges. One might think that one can simply replace the directed edges with undirected gadgets. However, this does not seem possible. A very high level reason is that the triangle inequality in undirected graphs can be used in both directions. The directed edges in the prior constructions were used to make sure that some pairs of vertices have short paths between them, while leaving the possibility of having large distances between other pairs. If undirected edges (or even gadgets) are used instead however, the triangle inequality implies short paths for pairs of vertices that the construction wants to avoid. A short path from $u$ to $v$ and a short path from $x$ to $v$ does imply a short path from $u$ to $x$ in undirected graphs, but not in directed graphs. This simple reason is basically why no simple extensions of the results of \cite{DW21} and \cite{Li21} to undirected graphs seem to work. (See Section~\ref{sec:proofsketch} for more about this.)

The fact that the triangle inequality can be used in both directions in undirected graphs, makes it difficult to extend the lower bound constructions to undirected graphs, but it also seems to make more algorithmic tradeoffs possible for undirected than for directed graphs, as evident from the Cairo, Grossi, Rizzi \cite{CGR} algorithms. It thus seems possible that a better than $2$ approximation algorithm running in linear time could be possible for undirected graphs.

The main result of this paper is a delicate construction that achieves the same tradeoff lower bounds for diameter in undirected graphs as the ones in directed graphs, thus showing that undirected diameter is just as hard. Namely:

\begin{theorem}
  Assuming SETH, for all integers $k\ge 2$, for all $\varepsilon>0$, a $(2-\frac{1}{k}-\varepsilon)$-approximation of Diameter in \emph{unweighted, undirected} graphs on $m$ edges requires $m^{1+1/(k-1)-o(1)}$ time. 
\label{thm:main}
\end{theorem}
Theorem~\ref{thm:main} was proved previously for $k=2$ \cite{rv13}, $k=3$ \cite{BRSVW18,Li20}, and $k=4$ \cite{Bonnet21ic}. The theorem is stated in terms of the number of edges $m$; our lower bound constructions are for the special case when $m=n^{1+o(1)}$ (i.e. very sparse graphs).

The main consequence of our theorem is that under SETH, there can be no better near-linear time approximation algorithm for undirected unweighted diameter than the simple $2$-approximation algorithm that runs BFS from an arbitrary vertex.

\paragraph{Outline}

In Section~\ref{sec:prelim}, we give some preliminaries for our construction.
In Section~\ref{sec:k4} we show how to prove Theorem~\ref{thm:main} for small cases $k=4$ and $k=5$ to illustrate some of our ideas.
We (re)prove Theorem~\ref{thm:main} for $k=4$, giving a simplified proof of Bonnet's result, and show how the proof can be modified to give a proof for $k=5$. The full proof for $k=5$ is deferred to Appendix~\ref{app:k5}.
In Section~\ref{sec:proofsketch}, we highlight some of the ideas in the construction.
Afterwards, we prove our formal results.
In Section~\ref{sec:all}, we prove Theorem~\ref{thm:main} in full generality.

%% file: prelim.tex
\section{Preliminaries} 
\label{sec:prelim}

\noindent For a positive integer $a$, let $[a]=\{1,2,\dots,a\}$.

\paragraph{$k$-OV.} A $k$-OV instance $\Phi$ is a set $A\subseteq\{0,1\}^{\mathbb{d}}$ of $n$ binary vectors of dimension $\mathbb{d}=\theta{(\log{n})}$ and the $k$-OV problem asks if we can find $k$ vectors $a_1,\ldots,a_k\in A$ such that they are orthogonal, i.e. $a_1\cdot \ldots \cdot a_k=0$. The \textit{$k$-OV Hypothesis} says that solving $k$-OV requires $n^{k-o(1)}$ time, and it is implied by SETH \cite{W04,TCS05}. 

Now we give the definitions that we use for our construction. At a very high level, we are going to start from a $k$-OV instance and create a diameter instance. To do so, we are going to make a graph where each node is a ``configuration", which we are going to define later. Each configuration consists of a number of ``stacks", where each stack has some of the vectors of the $k$-OV instance. There are relationships between different stacks in a configuration, and we define those relationships using ``coordinate arrays". Below we define these notions more formally.

\paragraph{Stacks.}
Given a $k$-OV instance $A\subset\{0,1\}^{\mathbb{d}}$, we make the following definitions.
A \emph{stack} $S=(a_1,\dots,a_{|S|})$ is a vector of elements of $A$ whose \emph{length} $|S|$ is a nonnegative integer.
Denote $a_1$ as the \emph{bottom} element of the stack and $a_{|S|}$ as the \emph{top} element of the stack.
We let $()$ denote the \emph{empty stack}, i.e., a stack with 0 vectors.
Given a stack $S=(a_1,\dots,a_\ell)$, a \emph{substack} $S_{\le \ell'}=(a_1,\dots,a_{\ell'})$ is given by the bottom $\ell'$ vectors of $S$, where $\ell'\le \ell$.
We call these tuples \emph{stacks}, because of the following operations.
The stack $popped(S)$ is the stack $(a_1,\dots,a_{\ell-1})$, i.e., the stack $S$ with the top element removed.
For a vector $b\in A$ and a stack $S=(a_1,\dots,a_\ell)$, the stack $S+b$ is the stack $S+b=(a_1,\dots,a_\ell,b)$.
The use of stacks as a primitive in our construction is motivated in Section~\ref{sec:proofsketch}.

\paragraph{Coordinate arrays.}
\begin{definition}
  \label{def:coord-1}
A \emph{$k$-coordinate-array $x$} is an element of $[\mathbb{d}]^{k-1}$. 
\end{definition}
In the reduction from $k$-OV, we only consider $k$-coordinate arrays, so we omit $k$ when it is understood.
For a $k$-coordinate array $x\in[\mathbb{d}]^{k-1}$ and an integer $\ell\in[k-1]$, let $x[\ell]$ denote the $\ell$th coordinate in the coordinate array $x$.
Also for a coordinate $c$ and a vector $a\in A$, $a[c]$ is the $c$th coordinate of $a$.
We index coordinate arrays by $x[\ell]$ and vectors in $A$ by $a[c]$, rather than $x_\ell$ and $a_c$ (respectively), for clarity. For a set of non-orthogonal vectors $\{a_1,\ldots,a_s\}$ for $s\le k$, let $ind(a_1,\ldots,a_s)$ return a coordinate $c$ such that $a_i[c]=1$ for all $i=1,\ldots,s$.

\begin{definition}[Stacks satisfying coordinate arrays]
\label{def:coord-2}
Let $S=(a_1,\ldots,a_{s})$ be a stack where $|S|\le k-1$, and let $x\in[\mathbb{d}]^{k-1}$ be a $k$-coordinate array.
We say that $S$ \emph{satisfies} $x$ if there exists sets $[k-1]= I_1\supset\cdots\supset I_s$ such that, for all $h=1,\dots,s$, we have $|I_h| = k-h$ and $a_h[x[i]] = 1$ for all $i\in I_h$. 
\end{definition}

\begin{table}
  \begin{minipage}[t]{0.24\textwidth}
  \vspace{0pt}
  \begin{tabular}{c|ccc}
    &$x[1]$ & $x[2]$ & $x[3]$\\ \hline
    $a_1$ & $1$ & $1$ & $1$ \\
    $a_2$ & $1$ & $1$ &     \\
    $a_3$ & $1$ &  & \\
  \end{tabular}
  \end{minipage}
  \begin{minipage}[t]{0.24\textwidth}
  \vspace{0pt}
  \begin{tabular}{c|ccc}
    & $x[1]$ & $x[2]$ & $x[3]$\\ \hline
    $a_1$ & $1$ & $1$ & $1$ \\
    $a_2$ & & $1$ & $1$     \\
    $a_3$ & & $1$ &   \\
  \end{tabular}
  \end{minipage}
  \begin{minipage}[t]{0.24\textwidth}
  \vspace{0pt}
  \begin{tabular}{c|ccc}
    & $x[1]$ & $x[2]$ & $x[3]$\\ \hline
    $a_1$ & $1$ & $1$ & $1$ \\
    $a_2$ & $1$ & $1$ &     \\
  \end{tabular}
  \end{minipage}
  \begin{minipage}[t]{0.24\textwidth}
  \vspace{0pt}
  \begin{tabular}{c|ccc}
    & $x[1]$ & $x[2]$ & $x[3]$\\ \hline
    $a_1$ & $1$ & $1$ & $1$ \\
    $a_2$ & & $1$ & $1$     \\
  \end{tabular}
  \end{minipage}
  \caption{In each of the above, $x=(x[1],x[2],x[3])$ is a 4-coordinate array. The left two tables depict that stack $(a_1,a_2,a_3)$ satisfies $x$, and the right two tables depict that stack $(a_1,a_2)$ satisfies $x$.}
\end{table}

\begin{lemma}
  \label{lem:stack-1}
  If stack $S$ satisfies a coordinate array $x$, then any substack of $S$ satisfies $x$ as well.
\end{lemma}
\begin{proof}
  This follows from the definition of satisfiability.
\end{proof}

\begin{lemma}
  \label{lem:stack-2}
  Let $S=(a_1,\dots,a_{|S|})$ and $S'=(b_1,\dots,b_{|S'|})$ be stacks, each with at most $k-1$ vectors from $A$, the $k$-OV instance, such that any $k$ vectors from among $a_1,\dots,a_{|S|},b_1,\dots,b_{|S'|}$ are not orthogonal.  Then there exists a coordinate array $x$ such that $S$ and $S'$ both satisfy $x$.
\end{lemma}
\begin{proof}
  By Lemma~\ref{lem:stack-1}, it suffices to prove this in the case that $|S|=|S'|=k-1$.
  Then $S=(a_1,\dots,a_{k-1})$ and $S'=(b_1,\dots,b_{k-1})$.
  Let $x[\ell]=ind(a_1,\dots,a_{k-\ell},b_1,\dots,b_\ell)$.
  Then for all $h=1,\dots,k-1$, we have $a_h[x[\ell]]=1$ for $\ell\le k-h$, so $S$ satisfies $x$ with sets $I_h = \{1,\dots,k-h\}$.
  Additionally, for all $h=1,\dots,k-1$, we have $b_h[x[\ell]]=1$ for $\ell=h,\dots,k-1$ so $S'$ satisfies $x$ with sets $I_h = \{h,\dots,k-1\}$.
  Hence, both $S$ and $S'$ satisfy $x.$ 
\end{proof}

\begin{lemma}\label{lem:yes-1}
  Let $a_1,\dots,a_k$ be $k$ orthogonal vectors.
  Suppose that $j$ is an index, $x$ is a coordinate array and $S=(a_1,\dots,a_j)$ and $S'=(a_k,\dots,a_{j+1})$ are two stacks.
  Then stacks $S$ and $S'$ cannot satisfy $x$ simultaneously. 
\end{lemma}
\begin{proof}
  Suppose for contradiction that $S$ and $S'$ both satisfy $x$.
  Let $[k-1]= I_1\supset I_2\supset\cdots \supset I_j$ be the sets showing that stack $S$ satisfies coordinate array $x$, and let $[k-1]= I_k\supset \cdots \supset I_{j+1}$ be the sets showing that stack $S'$ satisfies coordinate array $x$.
  Here, $I_j$ has size $k-j$ and $I_{j+1}$ has size $j$. 
  We have $|I_j\cap I_{j+1}| = |I_j| + |I_{j+1}|- |I_j\cup I_{j+1}| =k - |I_j\cup I_{j+1}| > 0$.
  Then $I_1\cap I_2\cap \cdots \cap I_k = I_j \cap I_{j+1} \neq \emptyset$, so there exists some $i\in I_1\cap \cdots\cap I_k$.
  For this $i$, we have $a_1[x[i]]=a_2[x[i]]=\cdots=a_k[x[i]] = 1$, so $a_1,\dots,a_k$ are not orthogonal, a contradiction.
  Thus, stacks $S$ and $S'$ cannot satisfy $x$ simultaneously. 
\end{proof}

%% file: k4.tex
\section{Main theorem for $k=4$}
\label{sec:k4}

In this section, we prove Theorem~\ref{thm:main} for $k=4$.
Theorem~\ref{thm:main} for $k=4$ was previously proven by Bonnet~\cite{Bonnet21ic}. 
Here we present a simpler proof that also illustrates some ideas in our general construction.
Furthermore, our construction for $k=4$ can be easily modified to give a hardness construction that proves Theorem~\ref{thm:main} for $k=5$.
We point out how this can be done in the $k=4$ construction below.
Since the modification is simple, and the proof of correctness is similar but more involved, we defer the full proof of the $k=5$ construction to Appendix~\ref{app:k5}, which can be read for more intuition for the main construction.
For two stacks $S=(a_1,\dots,a_s)$ and $T=(b_1,\dots,b_t)$, let $S\circ T$ denote the stack $(a_1,\dots,a_s,b_1,\dots,b_t)$.
\begin{theorem}
   Assuming SETH, for all $\varepsilon>0$, a $(\frac{7}{4}-\varepsilon)$-approximation of Diameter in \emph{unweighted, undirected} graphs on $m$ edges needs $m^{4/3-o(1)}$ time.
   \label{thm:74}
\end{theorem}
\begin{proof}
Start with a 4-OV instance $\Phi$ given by a set $A\subset \{0,1\}^{\mathbb{d}}$ with $|A|=n_{OV}$ and $\mathbb{d} = \theta(\log n_{OV})$.
We show how to solve $\Phi$ using an algorithm for Diameter.
First check in time $O(n_{OV}^3)$ whether there are three orthogonal vectors in $A$. If so, we know that $\Phi$ also has 4 orthogonal vectors, as we can add an arbitrary fourth vector to the triple and obtain a $4$-OV solution.

Thus, let us assume that there are no three orthogonal vectors.
We construct a graph with $\tilde O(n_{OV}^3)$ vertices and edges from the 4-OV instance such that (1) if $\Phi$ has no solution, any two vertices are at distance 4, and (2) if $\Phi$ has a solution, then there exists two vertices at distance 7.
Any $(7/4-\varepsilon)$-approximation for Diameter distinguishes between graphs of diameter 4 and 7.
Since solving $\Phi$ needs $n_{OV}^{4-o(1)}$ time under SETH, a $7/4-\varepsilon$ approximation of diameter needs $n^{4/3-o(1)}$ time under SETH.

\begin{figure}
    \centering
    \includegraphics[width=\linewidth]{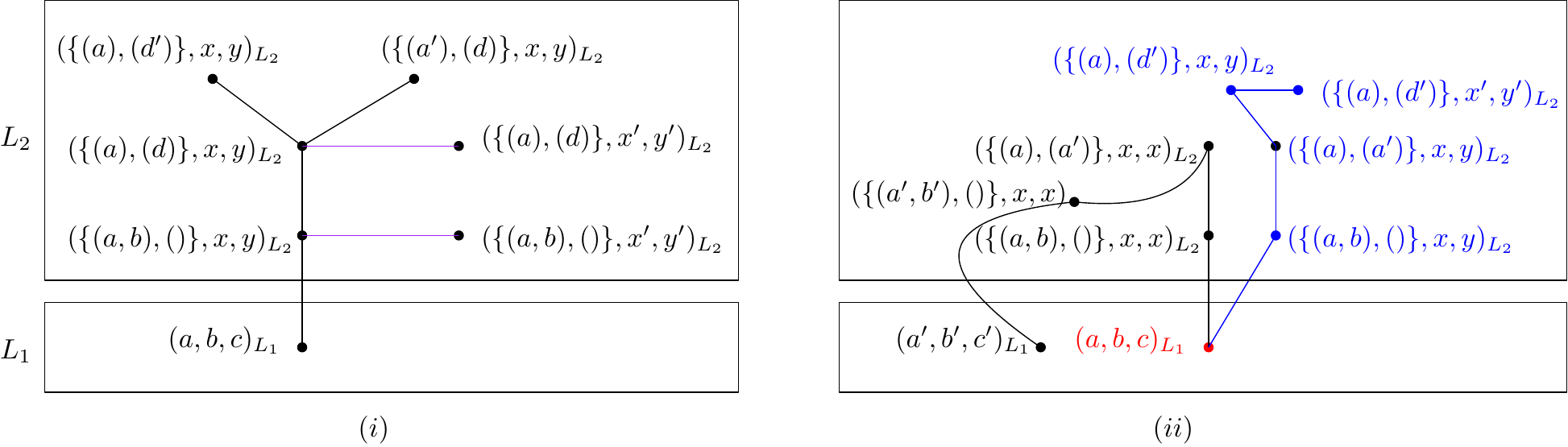}
    \caption{$(i)$ $4$-OV reduction graph. The purple edges are coordinate change edges. 
    $(ii)$ Paths in the first two cases of the NO case. Black path is for the case where both vertices are in $L_1$, blue path is for the case where one vertex is in $L_1$ and the other is in $L_2$ with two stacks of size $1$.}
    \label{fig:4-OV}
\end{figure}

\paragraph{Construction of the graph}
The graph $G$ is illustrated in Figure~\ref{fig:4-OV}(i) and constructed as follows.

The vertex set $L_1\cup L_2$ is defined on
\begin{align}
  L_1 &= \{S: \text{$S$ is a stack with $|S|=3$}\},  \nonumber\\
  L_2 &= \big\{(\{S_1,S_2\},x,y): \text{$S_1,S_2$ are stacks with $|S_1|+|S_2|=2$, }  \nonumber\\
   & \qquad\qquad\qquad\qquad\qquad \text{$x,y\in[\mathbb{d}]^3$ are coordinate arrays such that} \nonumber\\
   & \qquad\qquad\qquad\qquad\qquad \text{$S_1\circ S_2$ satisfies $x$ and $S_2\circ S_1$ satisfies $y$, OR} \nonumber\\
   & \qquad\qquad\qquad\qquad\qquad \text{$S_1\circ S_2$ satisfies $y$ and $S_2\circ S_1$ satisfies $x$}\big\}.   
\end{align}

In vertex subset $L_2$, the notation $\{S_1,S_2\}$ denote an \emph{unordered} set of two stacks.
As shown in Figure~\ref{fig:4-OV}, the vertices in $L_2$ are of two types: $(\{(a),(b)\},x,y)_{L_2}$ and $(\{(a,b),()\},x,y)_{L_2}$ for $a,b\in A$, $x,y\in [d]$.

Throughout, we identify tuples $(a,b,c)$ and $(\{S_1,S_2\},x,y)$  with vertices of $G$, and we denote vertices in $L_1$ and $L_2$ by $(a,b,c)_{L_1}$ and $(\{S_1,S_2\},x,y)_{L_2}$ respectively.
The (undirected unweighted) edges are the following.
\begin{itemize}
\item ($L_1$ to $L_2$) Edge between $(S)_{L_1}$ and $(\{\popped(S),()\},x,y)_{L_2}$ if stack $S$ satisfies both $x$ and $y$.

\item (vector change in $L_2$, type 1) For some vector $a\in A$ and stacks $S_1, S_2$ with $|S_1|\ge 1$, an edge between $(\{S_1,S_2\},x,y)_{L_2}$ and $(\{\popped(S_1),S_2+a\},x,y)_{L_2}$ if both vertices exist. 

In particular, as Figure~\ref{fig:4-OV} shows, $(\{(a,b),()\},x,y)_{L_2}$ has an edge to $(\{(a),(b')\},x,y)_{L_2}$ if both vertices exist. These are the only type 1 edges.

\item (vector change in $L_2$, type 2) For some vector $a\in A$ and stacks $S_1, S_2$ with $|S_1|\ge 1$, an edge between $(\{S_1,S_2\},x,y)_{L_2}$ and $(\{\popped(S_1)+a,S_2\},x,y)_{L_2}$ if both vertices exist.

In particular, as Figure~\ref{fig:4-OV} shows, $(\{(a),(b)\},x,y)_{L_2}$ has edges to $(\{(a'),(b)\},x,y)_{L_2}$ and $(\{(a),(b')\},x,y)_{L_2}$ if the vertices exist. These are the only type 2 edges.

\item (coordinate change in $L_2$) Edge between $(\{S_1,S_2\},x,y)_{L_2}$ and $(\{S_1,S_2\},x',y')_{L_2}$ if both vertices exist.
\end{itemize}

There are at most $n_{OV}^3$ vertices in $L_1$ and at most $n_{OV}^2\mathbb{d}^6$ vertices in $L_2$.
Note that each vertex of $L_1$ has $O(\mathbb{d}^2)$ neighbors, each vertex of $L_2$ has $O(n_{OV}+\mathbb{d}^2)$ neighbors.
The total number of edges and vertices is thus $O(n_{OV}^3\mathbb{d}^2)=\tilde O(n_{OV}^3)$.
We first show below how to change this construction for $k=5$, and then we show that the construction for $k=4$ has diameter 4 when $\Phi$ has no solution and diameter at least 7 when $\Phi$ has a solution.

\paragraph{Modifications for $k=5$.}
The construction of the Diameter instance $G$ when $k=5$ is very similar.
We instead start with a 5-OV (rather than 4-OV) instance $A\subset\{0,1\}^d$, and use the exact same graph, except the nodes in $L_1$ have a stack of size $4$ (rather than 3), and the total sizes of the stacks in $L_2$ is 3 (rather than 2). The descriptions of the edges are exactly the same.
We defer the proof of correctness of this construction for $k=5$ to Appendix~\ref{app:k5}. 
It is similar to the proof for $k=4$, but is more involved.

\paragraph{4-OV no solution}
Assume that the 4-OV instance $A\subset\{0,1\}^{\mathbb{d}}$ has no solution, so that no four (or three or two) vectors are orthogonal.
We show that any pair of vertices have distance at most 4, by casework:
\begin{itemize}
\item \textbf{Both vertices are in $L_1$:} Let the vertices be $(a,b,c)_{L_1}$ and $(a',b',c')_{L_1}$. 
By Lemma~\ref{lem:stack-2} there exists coordinate array $x$ satisfied by both stacks $(a,b,c)$ and $(a',b',c')$.
We claim the following is a valid path (see Figure \ref{fig:4-OV}ii):
\begin{align}
  (a,b,c)_{L_1} - (\{(a,b),()\},x,x)_{L_2}-(\{(a),(a')\},x,x)_{L_2}-(\{(a',b'),()\},x,x)_{L_2}-(a',b',c')_{L_1} 
\end{align}
The first edge and second vertex are valid because $(a,b,c)$ satisfies $x$ (and thus, by Lemma~\ref{lem:stack-1}, stack $(a,b)$ satisfies $x$).
By the same reasoning the last edge and fourth vertex are valid.
The third vertex is valid because each of $a$ and $a'$ have a 1 in all coordinates of $x$, so both $(a,a')$ and $(a',a)$ satisfy $x$.

\item \textbf{One vertex is in $L_1$ and the other vertex is in $L_2$ with two stacks of size 1:} Let the vertices be $(a,b,c)_{L_1}$ and $(\{(a'),(d')\},x',y')_{L_2}$.
By Lemma~\ref{lem:stack-2}, there exists a coordinate array $x$ satisfied by both stacks $(a,b,c)$ and $(a',d')$, and there exists a coordinate array $y$ satisfied by both stacks $(a,b,c)$ and $(d',a')$.
We claim the following is a valid path (see Figure \ref{fig:4-OV}ii):
\begin{align}
  (a,b,c)_{L_1}
  &-(\{(a,b),()\},x,y)_{L_2} \nonumber\\
  &-(\{(a),(a')\},x,y)_{L_2} \nonumber\\
  &-(\{(d'),(a')\},x,y)_{L_2}
  -(\{(a'),(d')\},x',y')_{L_2}.
\end{align}
The first edge and second vertex are valid because $(a,b,c)$ satisfies $x$ and $y$.
Vector $a$ has a one in each coordinate of $x$ and $y$, and stack $(a',d')$ satisfies $x$ and stack $(d',a')$ satisfies $y$, so stack $(a',a)$ satisfies $x$ and stack $(a,a')$ satisfies $y$, so the third vertex $(\{(a),(a')\},x,y)_{L_2}$ is valid, and thus the second edge is also valid.
The fourth vertex is valid because $(a',d')$ satisfies $x$ and $(d',a')$ satisfies $y$ by construction of coordinate arrays $x$ and $y$, and thus the third and fourth edges are valid.
Hence, this is a valid path.

\item \textbf{Both vertices are in $L_2$ with two stacks of size 1:}
Let the vertices be $(\{(a),(d)\},x,y)_{L_2}$ and $(\{(a'),(d')\},x',y')_{L_2}$.
Let $z_1\in[\mathbb{d}]$ be a coordinate where $a,d,a',d'$ are all 1, and let $z=(z_1,z_1,z_1)$ be a coordinate array.
Then the following is a valid path:
\begin{align}
  (\{ (a),(d)\},x,y)_{L_2}
  &-(\{(a),(d)\},z,z)_{L_2} \nonumber\\
  &-(\{(a'),(d)\},z,z)_{L_2} \nonumber\\
  &-(\{(a'),(d')\},z,z)_{L_2}
  -(\{(a'),(d')\},x',y')_{L_2}.
\end{align}
Indeed it's easy to check that any stack of two of $a,d,a',d'$ satisfies $z$, so all the vertices are valid and thus all the edges are valid, so this is a valid path.

\item \textbf{One vertex is in $L_2$ with two stacks of size 2 and 0:} 
For every vertex $u=(\{(a,b),()\},x,y)_{L_2}$ in $L_2$ with two stacks of size 2 and 0, any vertex of the form $v=(a,b,c)_{L_1}$ in $L_1$ has the property that the neighborhood of $u$ is a superset of the neighborhood of $v$ (by considering coordinate change edges from $u$).
Thus, any vertex that $v$ can reach in 4 edges can also be reached by $u$ in 4 edges.
In particular, since any two vertices in $L_1$ are at distance at most 4, any vertex in $L_1$ is distance at most 4 from any vertex in $L_2$ with two stacks of size 2 and 0.
Applying a similar reasoning, any vertex in $L_2$ with two stacks of size 2 and 0 is distance at most 4 from any vertex in $L_2$ with two stacks of size 2 and 0, and any vertex in $L_2$ with two stacks of size 1.
\end{itemize}
We have thus shown that any two vertices are at distance at most 4, proving the diameter is at most 4.

\paragraph{4-OV has solution.}
Now assume that the 4-OV instance has a solution.
That is, assume there exists $a_1,a_2,a_3,a_4\in A$ such that $a_1[i]\cdot a_2[i]\cdot a_3[i]\cdot a_4[i] = 0$ for all $i$.
Since there are no 3 orthogonal vectors, vectors $a_1,a_2,a_3,a_4$ are pairwise distinct.

Suppose for contradiction there exists a path of length at most 6 from $u_0=(a_1,a_2,a_3)_{L_1}$ to $u_6=(a_4,a_3,a_2)_{L_1}$.

Since all vertices in $L_2$ have self-loops with trivial coordinate-change edges, we may assume the path has length exactly 6.
Let the path be 
$u_0=(a_1,a_2,a_3)_{L_1}, u_1,\dots,u_6=(a_4,a_3,a_2)_{L_1}$.
We may assume the path never visits $L_1$ except at the ends: if $u_i=(S)\in L_1$, then $u_{i-1}=(\{\popped(S),()\},x,y)$ and $u_{i+1}=(\{\popped(S),()\},x',y')$ are in $L_2$, and in particular $u_{i-1}$ and $u_{i+1}$ are adjacent by a coordinate change edge, so we can replace the path $u_{i-1}-u_i-u_{i+1}$ with $u_{i-1}-u_{i+1}-u_{i+1}$, where the last edge is a self-loop.

For $i=1,2,3$, let $p_i$ denote the largest index such that vertices $u_0,u_1,\dots,u_{p_i}$ all contain a stack that has $(a_1,\dots,a_i)$ as a substack.
Because we never revisit $L_1$, we have $p_3 = 0$.
For $i=1,2,3$, let $q_i$ be the smallest index such that vertices $u_{q_i},\dots,u_6$ all contain a stack with $(a_4,\dots,a_{5-i})$ as a substack.
Because we never revisit $L_1$, we have $q_3 = 6$.
We show that, 
\begin{claim}
For $i=1,2,3$, between vertices $u_{p_i}$ and $u_{q_{4-i}}$, there must be a coordinate change edge.
\label{cl:4-yes}
\end{claim}
\begin{proof}
Suppose for contradiction there is no coordinate change edge between $u_{p_i}$ and
$u_{q_{4-i}}$. We show a contradiction for each of $i=1,2,3$.

First, consider $i=3$.
Here, $u_{p_i}=u_0 = (a_1,a_2,a_3)_{L_1}$.
By minimality of $q_1$, vertex $u_{q_1}$ is of the form $(\{(e), (a_4)\},x,y)_{L_2}$ 
for some vector $e$.
Then stack $(a_4)$, satisfies one of $x$ and $y$.
Since there is no coordinate change edge between $u_0$ and $u_{q_1}$, we must have 
$u_1 = (\{(a_1,a_2),()\},x,y)$  for the same coordinate arrays $x$ and $y$, so stack $(a_1,a_2,a_3)$ satisfies both $x$ and $y$.
Hence, there is some coordinate array satisfied by both $(a_1,a_2,a_3)$ and $(a_4)$, which is a contradiction of Lemma~\ref{lem:yes-1}.
By a similar argument, we obtain a contradiction if $i=1$.

Now suppose $i=2$.
By maximality of $p_2$, vertex $u_{p_2}$ is of the form $(\{(a_1,a_2),()\},x,y)_{L_2}$.
By minimality of $q_2$, vertex $u_{q_2}$ is of the form $(\{(a_4,a_3),()\},x,y)_{L_2}$.
The coordinate arrays $x$ and $y$ are the same between the two vertices because there is
no coordinate change edge between them by assumption.
Then stacks $(a_1,a_2)$ and $(a_4,a_3)$ satisfy both coordinate arrays $x$ and $y$, which contradicts Lemma~\ref{lem:yes-1}.
\end{proof}

Since coordinate change edges do not change any vectors, by maximality of $p_i$, the edge $u_{p_i}u_{p_{i}+1}$ cannot be a coordinate change edge for all $i=1,2,3$.
Similarly, by minimality of $q_i$, the edge $u_{q_{i}-1}u_{q_i}$ cannot be a coordinate change edge for all $i=1,2,3$.
Consider the set of edges
\begin{align}
  \label{eq:4-yes-1}
  u_{p_3}u_{p_3+1}, u_{p_2}u_{p_2+1}, u_{p_1}u_{p_1+1}, 
  u_{q_1-1}u_{q_1}, u_{q_2-1}u_{q_2}, u_{q_3-1}u_{q_3}.
\end{align}
By the above, none of these edges are coordinate change edges.
These edges are among the 6 edges $u_0u_1,\dots,u_5u_6$.
Additionally, the edges $u_{p_i}u_{p_i+1}$ for $i=1,2,3$ are pairwise distinct, and the edges $u_{q_i-1}u_{q_i}$ for $i=1,2,3$ are pairwise distinct.
Edge $u_{p_3}u_{p_3+1}$ cannot be any of $u_{q_i-1}u_{q_i}$ for $i=1,2,3$, because we
assume our orthogonal vectors $a_1,a_2,a_3,a_4$ are pairwise distinct and $u_{p_3-1}=u_1$ does not
have any stack containing vector $a_4$.
Similarly, $u_{q_3-1}u_{q_3}$ cannot be any of $u_{p_i}u_{p_i+1}$ for $i=1,2,3$.
Thus, the edges in \eqref{eq:5-yes-1} have at least 4 distinct edges, so our path has at most 2 coordinate change edges.
By Claim~\ref{cl:4-yes}, there must be at least one coordinate change edge.
We now do casework on the number of coordinate change edges.

\textbf{Case 1: the path $u_0,\dots,u_6$ has one coordinate change edge.}
By Claim~\ref{cl:4-yes}, since vertex $u_{p_3}=u_0$ is before the coordinate change edge, edge $u_{q_1-1}u_{q_1}$ must be after the coordinate change edge, and similarly edge $u_{p_1}u_{p_1+1}$ must be before the coordinate change edge.
Then all of the edges in \eqref{eq:4-yes-1} are pairwise distinct, so then the path has 6 edges from \eqref{eq:4-yes-1} plus a coordinate change edge, for a total of 7 edges, a contradiction.

\textbf{Case 2: the path has two coordinate change edges.}
Again, by Claim~\ref{cl:4-yes}, for $i=1,2,3$, edges $u_{q_i-1}u_{q_i}$ must be after the first coordinate change edge, and edge $u_{p_i}u_{p_i+1}$ must be before the second coordinate change edge. 
Since we have 6 edges total, we have at most 4 distinct edges from \eqref{eq:4-yes-1}, so
there must be at least two pairs $(i,j)$ such that the edges $u_{p_i}u_{p_i+1}$ and $u_{q_j-1}u_{q_j}$ are equal.
By above this edge must be between the two coordinate change edges, so edges
$u_{p_2}u_{p_2+1}, u_{p_1}u_{p_1+1}, u_{q_2-1}u_{q_2}, u_{q_1-1}u_{q_1}$ are all between
the two coordinate change edges.
However, this means that vertices $u_{p_2}$ and $u_{q_2}$ are between the two coordinate
change edges, contradicting Claim~\ref{cl:4-yes}.

This shows that $(a_1,a_2,a_3)_{L_1}$ and $(a_4,a_3,a_2)_{L_1}$ are at distance at least 7, completing the proof.
\end{proof}

%% file: proofsketch.tex
\section{Overview of the general $k$ reduction}
\label{sec:proofsketch}

\subsection{The basic setup}
\label{ssec:proofsketch-0}
To prove Theorem~\ref{thm:main} in general, we start with a $k$-OV instance $A\subset\{0,1\}^{\mathbb{d}}$ with size $|A|=n_{OV}$ and dimension $\mathbb{d}\le O(\log n_{OV})$, and construct a graph $G$ on $\tilde O(n_{OV}^{k-1})$ vertices and edges such that, if the set $A$ has $k$ orthogonal vectors (Yes case), the diameter of $G$ is at least $2k-1$, and otherwise (No case) the diameter of $G$ is at most $k$.
Throughout, we refer to elements of $A$ as \emph{vectors} and elements of $[\mathbb{d}]$ as \emph{coordinates}.
%We may assume without loss of generality that the all ones vector $\mathbf{1}$ is in the set $A$.
%We may also assume without loss of generality that no $k-1$ vectors of $A$ are orthogonal, as this can be checked in $\tilde O(n_{OV}^{k-1})$ time, and if there were then we already know we have $k$ orthogonal vectors.
Each vertex of $G$ is identified by a \emph{configuration} $H$, which contains vectors (in $A$) and coordinates (in $[\mathbb{d}]$), along with some meta-data.
Vertices must be \emph{valid} configurations $H$, meaning vectors of $H$ have 1s in specified coordinates of $H$.
Edges between configurations of $G$ change up to one vector and/or some coordinates, and we think of edges as performing \emph{operations} on configurations.
We ensure the graph is undirected by choosing operations that are invertible.
\begin{figure}
\begin{center}
\includegraphics[width=15cm]{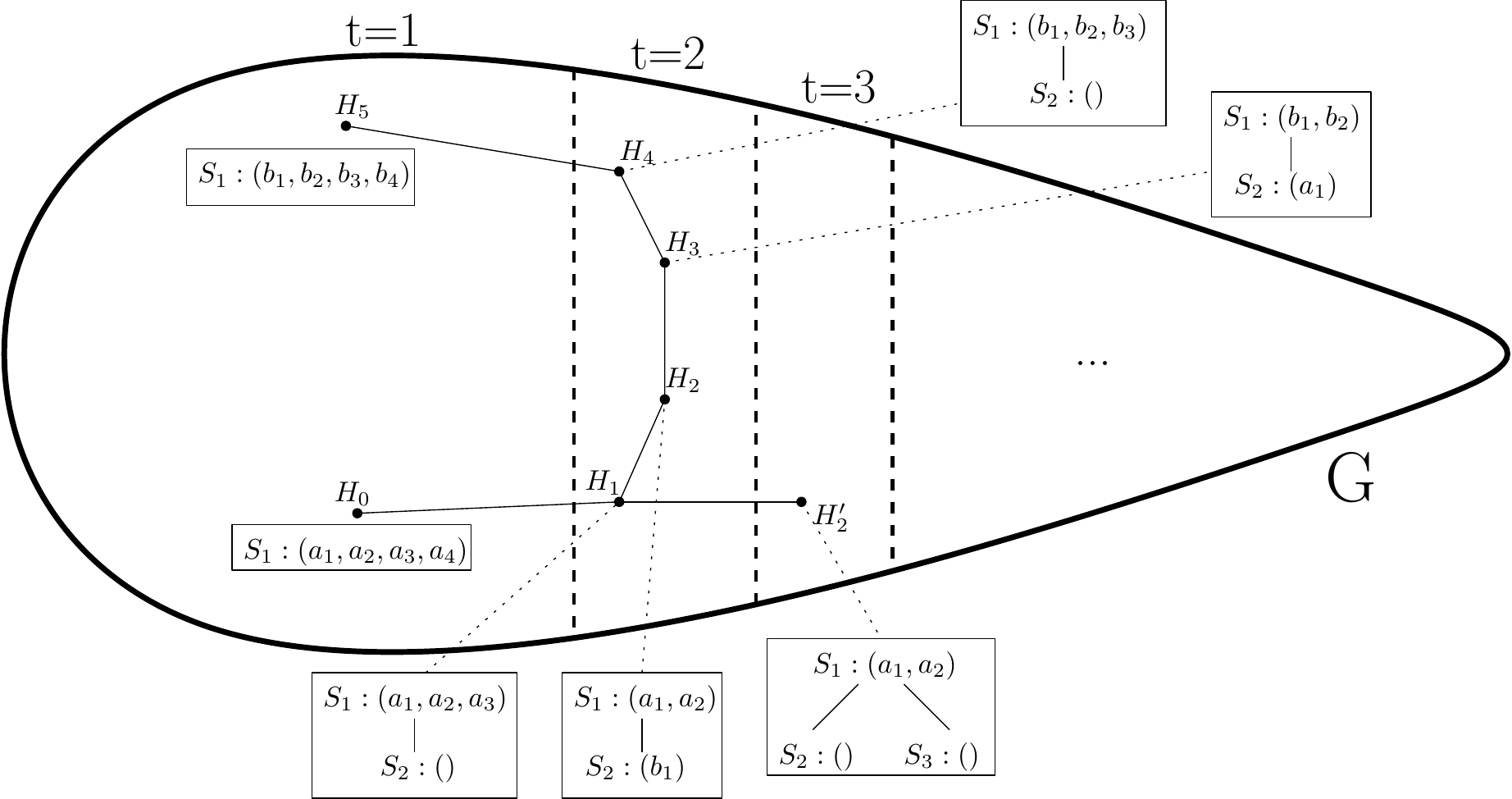}
\end{center}
\caption{Our Diameter instance $G$, illustrated for $k=5$. Vertices are \emph{configurations} and edges are \emph{operations} on configurations. Edges within configurations hold coordinate arrays (suppressed in the figure).}
\label{fig:bigpicture}
\end{figure}

\subsection{The Diameter instance construction}
\label{ssec:config}
We now sketch the definitions of configurations and operations, which define the vertices and edges, respectively, of the Diameter instance $G$.
Figure~\ref{fig:bigpicture} illustrates our graph $G$ and some vertices and edges.
\paragraph{Configurations.}
A \emph{configuration} $H$ is identified by the following:
\begin{enumerate}
\item A positive integer $t$ and a sequence of $t$ lists of vectors $S_1,\dots,S_t$, which we call \emph{stacks}.
Stack $S_1$ is special and is called the \emph{root}, and we require $S_1$ to have at least $(k-2)/2$ vectors.
\item A collection of $O(k^2)$ elements of $[\mathbb{d}]^{k-1}$, which we call \emph{coordinate arrays}, which are each tagged with one or two of the stacks $S_1,\dots,S_t$ (here, we omit the details of this tagging).
\end{enumerate}
The \emph{size} of a configuration is $t+\sum_{i=1}^{t} |S_i|$, the number of stacks plus the number of vectors.
The vertices of our Diameter instance $G$ correspond to the valid (defined below) size-$k$ configurations (see Figure~\ref{fig:bigpicture}).\footnote{Prior lower bounds \cite{BRSVW18, Bonnet21, DW21,Li21, Bonnet21ic} resemble this construction but with only $t\le 2$ stacks. Handling more than two stacks is nontrivial but seems necessary for our undirected, general-$k$ result.}

\paragraph{Valid configurations.}
%\begin{wrapfigure}{r}{5.5cm}
%\begin{center}
%\includegraphics[width=5cm]{proofsketch-0b.pdf}
%\end{center}
%\caption{Stacks satisfying coordinate arrays when $k=5$. Blue squares indicate a coordinate that must be 1. Left: Lemma~\ref{lem:stack-2}, there exists a coordinate array satisfied by stacks $(a_1,a_2,a_3,a_4)$ and $(b_1,b_2,b_3,b_4)$, Right: Lemma~\ref{lem:yes-1}, if stacks $(a_1,a_2,a_3)$ and $(a_5,a_4)$ satisfy a common coordinate array, vectors $a_1,\dots,a_5$ not orthogonal.}
%\label{fig:stack}
%\end{wrapfigure}

%A configuration is \emph{valid} if its stacks satisfy some constraints specified by its coordinate arrays.
%To specify these constraints we can define what it means for a stack to satisfy a coordinate-array (Definition~\ref{def:coord-1}, see Figure~\ref{fig:stack}). 
%This definition was implicit in prior constructions \cite{BRSVW18, DW21,Li21}:
%To specify what makes a configuration valid, we use the following definition, implicit in prior constructions \cite{BRSVW18, Bonnet21,DW21,Li21, Bonnet21ic}:
%Stack $S=(a_1,\dots,a_s)$ \emph{satisfies} coordinate array $x$ if there exists sets $[k-1]= I_1\supset\cdots\supset I_s$ such that, for all $h=1,\dots,s$, we have $|I_h| = k-h$ and $a_h[x[i]] = 1$ for all $i\in I_h$.
A configuration is \emph{valid} if every coordinate array is \emph{satisfied} (defined in Definition~\ref{def:coord-2}) by its one or two tagged stacks.
This technical notion of ``stacks satisfying coordinate arrays'', implicit in prior constructions \cite{BRSVW18, Bonnet21,DW21,Li21, Bonnet21ic}, has two key properties.
\begin{enumerate}
\item (Lemma~\ref{lem:stack-2}, used in No case) 
If every $k$ vectors among the vectors of stacks $S$ and $S'$ are not orthogonal, $S$ and $S'$ satisfy a common coordinate array.
%Let $(a_1,\dots,a_s)$ and $(b_1,\dots,b_{s'})$ be stacks with at most $k-1$ vectors such that any $k$ vectors from among $a_1,\dots,a_s,b_1,\dots,b_{s'}$ are orthogonal.
%Then there exists a coordinate array satisfied by both $(a_1,\dots,a_s)$ and $(b_1,\dots,b_{s'})$.
\item (Lemma~\ref{lem:yes-1}, used in Yes case) If stacks $(a_1,\dots,a_j)$ and $(a_k,\dots,a_{j+1})$ satisfy a common coordinate array, then $a_1,\dots,a_k$ are not orthogonal.
\end{enumerate}
%Each coordinate array $x$ in a configuration is tagged with one or two stacks, and 
%For each edge $S_1S_i$, we include the following coordinate arrays:
%\begin{itemize}
%\item $X_*$ satisfied by $S_1$ and $S_i$,
%\item for $j=2,\dots,k$, a coordinate array $X_{i,j}$ satisfied by $S_i$ and, if $j<i$, by $S_j$.
%\item for $j=2,\dots,k$, a coordinate array $X_{1,j}$ satisfied by $S_1$.
%\end{itemize}

%It may seem redundant that there are several coordinate arrays that need only be satisfied by $S_1$ and several that are only satisfied by $S_i$.

\paragraph{Operations.}
Operations (Figure~\ref{fig:op}) are composed of half-operations, which are one of the following.
\begin{enumerate}
\item (Vector insertion) Insert a vector at the end of a stack.
\item (Vector deletion) Delete the last vector of a stack.
\item (Node\footnote{We say ``Node insertion'' instead of ``stack insertion'' because in the actual construction, we place the stacks at nodes of a graph.} insertion) Insert an empty non-root stack.
%any of positions $2,\dots, t+1$, and add an edge from $S'$ to $S_1$ with associated coordinate arrays.
\item (Node deletion) Delete an empty non-root stack.\footnote{In the formal construction, we require that the deleted stack is either $S_{t-1}$ or $S_t$, and require an analogous condition for node insertions. The proof holds without this requirement, but it is a notational convenience in the proof of Lemma~\ref{lem:op-4}.}
%and delete the corresponding edge and coordinate arrays.
\item (Flip) If $t=2$, switch the two stacks $S_1$ and $S_2$, making $S_2$ the new root.
\end{enumerate}
Note, vectors are inserted and deleted ``First In Last Out'', hence the term ``stack'' (see \textbf{Why stacks?}).
\begin{figure}
\begin{center}
\includegraphics[height=3cm]{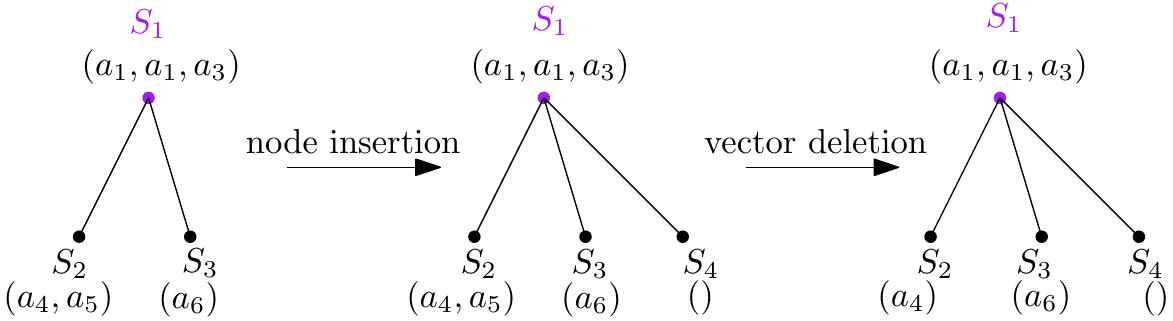}
\end{center}
\caption{(Full) Operation on configuration of size $k=7$. Root stack $S_1$ is in purple. Coordinate arrays (suppressed in figure) are attached to edges.}
\label{fig:op}
\end{figure}

During node insertion and deletions, we also insert and delete, respectively, coordinate arrays from the configuration.
Specifying how to do this is a significant challenge.
%A significant challenge is specifying how to insert and delete these coordinate arrays.
At a high level, we associate with each configuration a star graph\footnote{We emphasize there are now two types of graphs: the Diameter instance, and the star graphs of each configuration.} having vertices $S_1,\dots,S_t$ and edges $S_1S_i$ for $i=2,\dots,t$ (hence $S_1$ is called the root, see Figures~\ref{fig:bigpicture} and \ref{fig:op}).
We attach each coordinate array to an edge (the edge's endpoints may be different from the coordinate array's tagged stack(s)), and insert and delete coordinate arrays when their associated edge is inserted or deleted.
%For clarity, we emphasize there are now two types of graphs: the Diameter instance, whose vertices are configurations, and the star graphs in configurations, who nodes are stacks.

%With the node insertion/deletion and flip operations, the indices of the stacks change, so one needs to carefully to relabel the coordinate arrays so that they still constrain the same stack(s) (for example, in the flip, we swap $X_{2,j}$ with $X_{1,j}$).
%In practice, we avoid this relabelling by assigning separate, non-index labels to the stacks, that do not change with an operation, and index the coordinate arrays by these separate labels.

%\TODO{Need to specify how coordinate arrays are deleted and inserted with node insertions/deletions, }

A \emph{(full) operation} consists of two half-operations: a vector insertion or node insertion followed by a vector deletion or node deletion.
We also allow operations to include a flip operation after the half-operations.
To ensure at most $\tilde O(n_{OV}^{k-1})$ edges, we do not allow operations between two configurations with one stack ($t=1$).
An operation is \emph{valid} if the starting and ending configuration are valid.\footnote{\label{footnote:plus1}We also require validity of intermediate configurations after one of the two half-operations. In the Yes case, this gives an extra +2, proving the diameter is $2k-1$, rather than $2k-3$.}
The Diameter instance $G$ has the edge $(H,H')$ if applying a valid operation to $H$ gives $H'$.

\paragraph{Basic properties.}
We check a few basic properties of the construction.
\begin{itemize}
\item Operations leave the size $t+\sum_{i=1}^{t} |S_i|$ of a configuration invariant, so the edges are well-defined. (this is why we defined size as $t+\sum_{i=1}^{t} |S_i|$.)
%The graph vertices are all valid configurations of size $k$, applying any valid operation to a vertex gives the configuration of another vertex, so the edges are well-defined.
\item Since the Diameter instance deals with size $k$ configurations, each configuration has at most $k-1$ vectors, so there are at most $\tilde O(n_{OV}^{k-1})$ vertices. Similarly, one can check that there are $\tilde O(n_{OV}^{k-1})$ edges, and that the graph can be constructed in $\tilde O(n_{OV}^{k-1})$ time.
\item Operations are invertible, so the graph is undirected. 
%A vector insertion can be inverted by a vector deletion, a node insertion can be inverted by a node deletion, and two flips leave the configuration unchanged.
For example, a vector insertion/node deletion can be inverted by a node insertion/vector deletion. 
\end{itemize}

\paragraph{Why stacks?}
That is, why are vectors inserted ``First In Last Out'' from stacks?
Crucially, stacks ensure that $k-1$ operations are needed to delete the bottom vector of a configuration with one stack.
As in prior constructions, the Yes case shows that if $a_1,\dots,a_k$ are orthogonal, the one-stack configurations $H$ and $H'$ with stacks $(a_1,\dots,a_{k-1})$ and $(a_k,\dots,a_2)$ are at distance $2k-1$.
If we could delete $a_1$ from $H$ in less than $k-1$ operations, we could arrive in $k-2$ operations at a configuration $H''$ such that any $k$ vectors among $H''$ and $H'$ are not orthogonal.
Then $H''$ and $H'$ are distance $k$ by the No case, so $d(H,H') \le d(H,H'') + d(H'',H') \le 2k-2$ by the triangle inequality, a contradiction.

\subsection{Correctness}
We now sketch why $G$ has diameter at most $k$ in the No case and at least $2k-1$ in the Yes case.
\paragraph{No case.}
\begin{figure}
\begin{center}
\includegraphics[width=16cm]{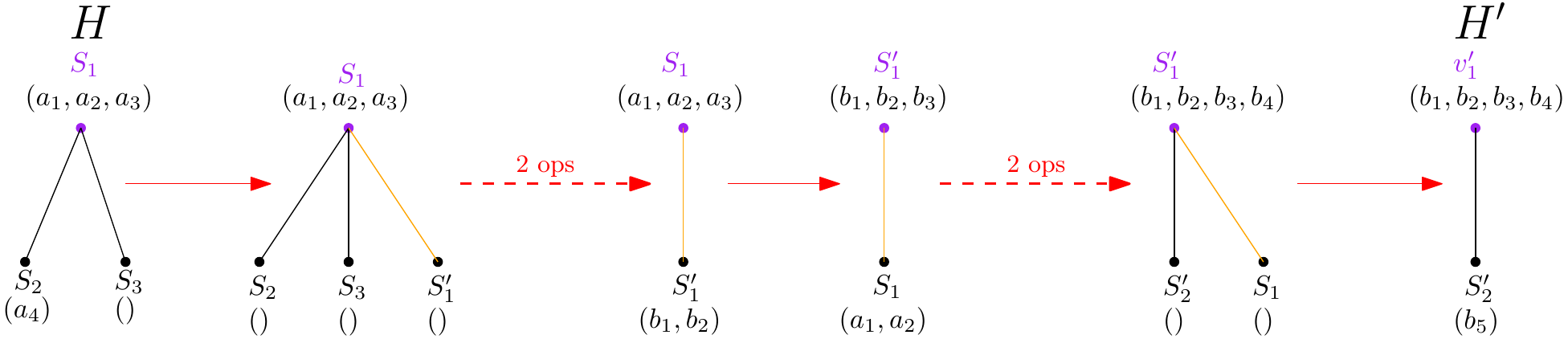}
\end{center}
\caption{No-case path between configurations $H$ and $H'$ for $k=7$. We delete all non-root stacks of $H$ before inserting any non-root stacks of $H'$. Orange edges hold auxiliary coordinate arrays not belonging to $H$ or $H'$.}
\label{fig:proofsketch-no}
\end{figure}
Suppose any $k$ vectors are not orthogonal.
We want to show we can reach any configuration $H'$ from any other configuration $H$ with $k$ valid operations.
If the operations do not need to be valid, this is easy:
insert the nodes and vectors of $H'$ while deleting the vectors and nodes from $H$.
We need $k$ deletions to remove $H$ (because it has size $k$), and $k$ insertions to build $H'$, so we pair the insertions and deletions to get from $H$ to $H'$ in $k$ full operations.

Since these operations may not all be valid, we must carefully choose the order of the insertions and deletions.
%Lemma~\ref{lem:stack-2} helps ensures we can choose a valid path.
%As we delete $H$ and insert $H'$, we introduce auxiliary coordinate arrays that are not a part of configurations $H$ or $H'$ but are inserted during a node insertion, and we need to delete these coordinate arrays before we reach the final configuration $H$.
The root stack is key in choosing the path.
Let $S_1$ and $S_1'$ be the root stacks of $H$ and $H'$.
Because $S_1$ and $S_1'$ each have at least $(k-2)/2$ vectors (by definition, and crucially), we can choose a path from $H$ to $H'$ that first deletes all the non-root stacks of $H$ while only adding stack $S_1'$ and its vectors (see Figure~\ref{fig:proofsketch-no}).
Then when $S_1'$ has at least $(k-2)/2$ vectors, we apply a flip operation, so that $S_1'$ is the new root, and we build the remainder of $H'$ while deleting stack $S_1$.\footnote{\label{footnote:path}
By viewing a path $H_1,H_2,\dots$ as a sequence of operations on $H_1$, we can naturally identify stacks and coordinates across different configurations in the path, talking about, for example, a stack $S_1$ of $H_1$ being in $H_i$.
For this overview, this informality suffices.
To avoid ambiguity in the formal proof, we give stacks a label that does not change between operations (and contract pairs of configurations that are equivalent up to permuting labels).
}

%This path works in part, because all non-root stacks of $H$ are deleted before any non-root stacks of $H'$ are inserted.
%Furthermore, Lemma~\ref{lem:stack-2} and orthogonality ensures we can choose coordinate arrays in intermediate configurations to be satisfied by the appropriate stacks belonging to $H$ and $H'$.
%Very roughly, this path works because it avoids ever (in the intermediate configurations) tagging a coordinate array ``belonging to'' $H$ with a stack of $H'$, and vice-versa.

Roughly, this path works because all coordinate arrays tagged with both a stack in $H$ and a stack in $H'$ are ``auxiliary'', belonging to neither $H$ nor $H'$; they are attached to $S_1S_1'$, the orange edges in Figure~\ref{fig:proofsketch-no}.
This requirement is necessary, as $H$ and $H'$ are generic configurations, so stacks of $H$ may not satisfy any coordinate array of $H'$ and vice-versa. 
Furthermore, Lemma~\ref{lem:stack-2} and non-orthogonality let us choose these auxiliary coordinate arrays to always be satisfied by their tagged stacks, making the path valid.

\paragraph{Yes case.}

Suppose that there are $k$ orthogonal vectors $a_1,\dots,a_k$.
We sketch why our graph $G$ has diameter at least $2k-3$. 
The formal proof shows the diameter is at least $2k-1$ (see footnote \ref{footnote:plus1}).

Let $H_0$ be the configuration with one stack $S_1=(a_1,\dots,a_{k-1})$, and let $H_{2k-4}$ be the configuration with one stack $S_1'=(a_k,\dots,a_2)$.
Suppose for contradiction there is a path $H_0,H_1,\dots,H_{2k-4}$.
At the highest level, we find two stacks $(a_k,\dots,a_{j+1})$ and $(a_1,\dots,a_j)$ from intermediate configurations satisfying a common coordinate array, contradicting Lemma~\ref{lem:yes-1}.

%If stacks $S_1$ and $S_1'$ ``correspond to'' (see footnote~\ref{footnote:path}) the same stack in $H_0,\dots,H_{2k-4}$, then $S_1$ must be emptied before $S_1'$ is built, costing $2k-2$ operations, which is too many.
%Thus, we may assume $S_1$ and $S_1'$ correspond to different stacks.

%If stack $S_1$ is never deleted, then the path deletes all vectors of $S_1$ and then inserts all vectors of $S_1'$, using $2k-2$ operations, which is a contradiction.
%First if $i=2k-2$ and $j=0$, i.e., the stack $S_1$ never gets deleted and just becomes $S_1'$, then we must first use $k-1$ vector deletions to $S_1$ to get the empty stack, and then $k-1$ vector insertions to build $S_1'$. Since full operations always apply the ``insertion'' half-operation before the ``deletion'' half-operation, these insertions and deletions are part of different half operations, so we must have at least $2k-2$ full operations, so the path must have length at least $2k-2$, contradiction.but

\begin{figure}
\begin{center}
\includegraphics[height=3cm]{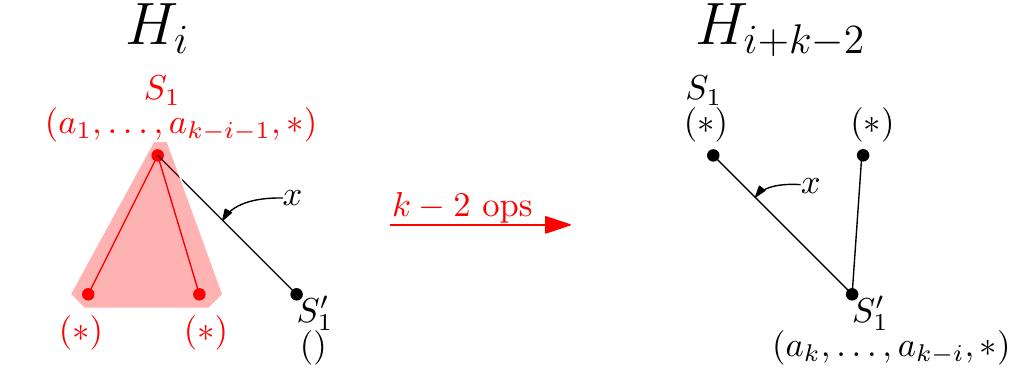}
\end{center}
\caption{The Yes case. We find a coordinate array $x$ satisfied by stack $(a_1,\dots,a_{k-i-1})$ in some configuration and satisfied by stack $(a_k,\dots,a_{k-i})$ in another configuration, contradicting Lemma~\ref{lem:yes-1}. 
Here, coordinate array $x$ is both attached to edge $S_1S_1'$ (so it is inserted and deleted with the edge) and tagged with stacks $S_1$ and $S_1'$ (so stacks $S_1$ and $S_1'$ need to satisfy $x$).
The *s represents some (possibly zero) vectors.}
\label{fig:proofsketch-yes}
\end{figure}

Let $i$ be the smallest index such that configurations $H_i,H_{i+1},\dots,H_{2k-4}$ all contain stack $S_1'$.
It is easy to check that $i\le k-3$ so $H_i$ also contains stack $S_1$.
For this sketch, assume that stacks $S_1$ and $S_1'$ are adjacent in the configuration $H_i$'s star graph.\footnote{There are two other cases: $S_1$ and $S_1'$ are the same stack, and $S_1$ and $S_1'$ are nonadjacent stacks in the star graph. The first case is easy, and the nonadjacent case is similar but more technical, depending on the details of tagging coordinate arrays with stacks.}
Since this star graph is always a tree, and valid operations can only delete leaf nodes, stack $S_1$ can only be deleted by deleting all of $H_i$ minus stack $S_1'$ (The red stacks/vectors in Figure~\ref{fig:proofsketch-yes}), which takes $k-1$ operations (Lemma~\ref{lem:yes-3}).
Thus, configurations $H_i,\dots,H_{i+k-2}$ all have stacks $S_1$ and $S_1'$ and the edge between them.

Our construction guarantees a coordinate array $x$ attached to edge $S_1S_1'$ that is satisfied by $S_1$ and $S_1'$.
Hence, $x$ is satisfied by $S_1$ and $S_1'$ in each of $H_i,\dots,H_{i+k-2}$. 
In $H_i$, stack $S_1$ must have a prefix of $(a_1,\dots,a_{(k-1)-i})$, which thus satisfies $x$. \footnote{If a stack satisfies coordinate array, its prefixes (substacks) also satisfy that coordinate array (Lemma~\ref{lem:stack-1}).}
In $H_{i+k-2}$, stack $S_1'$ must have a prefix of $(a_k,\dots,a_{(2k-4)-(i+k-2)+2})$, which also satisfies $x$.
Hence stacks $(a_1,\dots,a_{k-i-1})$ and $(a_k,\dots,a_{k-i})$ satisfy a common coordinate array, contradicting Lemma~\ref{lem:yes-1}.

%% file: simplified.tex
\section{The main theorem for general $k$}
\label{sec:all}

We describe below a reduction from $k$-OV to $2k-1$ vs. $k$ Diameter with time $O(n^{k/(k-1)})$ on graphs with edges of weight 1 or 0.
This immediately gives a reduction from $k$-OV to $2k-1$ vs. $k$ Diameter with time $O(n^{k/(k-1)})$ on unweighted graphs, by contracting the edges of weight 0.
For clarity of exposition, we describe the reduction to the 0/1-weighted graph.

Throughout the construction, fix $k'=\floor{k/2}+1$.
Throughout the construction, all coordinate arrays are $k$-coordinate arrays.
Let $\Phi$ be a $k$-OV instance given by a set $A$ of $n$ vectors of length $O(\log{n})$. We create a graph $G$ using this instance.
First we need a few definitions.

\subsection{Configurations}

\paragraph{Edge constraints.}
The vertices of our construction are ``configurations" which we are going to define formally later. 
Each configuration is a small graph, in which each vertex is assigned a stack and each edge puts constraints between those stacks. 
These \emph{edge-constraints} on edges are of the following form.
Recall that a coordinate array is an element of $[\mathbb{d}]^{k-1}$. 
\begin{definition}[Edge-constraint]
In a graph with an edge connecting vertices $v$ and $v'$, a \emph{$(v,v')$-edge-constraint $X$} (or \emph{edge-constraint} when $(v,v')$ is implicit) is a tuple of $2k'+1$ coordinate-arrays: $X_{v,i}$ and $X_{v',i}$ for $i\in[k']$, and $X_{*}$.
\end{definition}
We later define how these $2k'+1$ coordinate arrays of a $(v,v')$-edge-constraint relate with the stacks assigned to $v$ and $v'$, as well as the stacks of other vertices.

\paragraph{Configurations.}
With these edge-constraints defined, we can now define a configuration.
\begin{definition}[Configuration]
  A \emph{configuration} $H$ is an undirected star\footnote{Recall a star graph is a tree with a \emph{center} vertex adjacent to all other vertices.} graph $H$ with nodes $V(H)$ labeled by distinct elements of $[2k']$, such that 
  \begin{enumerate}
  \item The center node, denoted $\rho(H)$, of the star graph $H$ is called the \emph{root} (if the graph has two nodes, either one could be the root),
  \item $H$ is equipped with a total order $\prec_H$ on the vertices of $H$ such that the root is the smallest node of $\prec_H$,
  \item Each node $v$ of $H$ is assigned a stack $S_v(H)$, and
  \item Each edge $(v,v')$ of $H$ is labeled with an $(v,v')$-edge constraint $X^{v,v'}$. As graph $H$ is undirected, we equivalently denote $X^{v,v'}$ by $X^{v',v}$.
  \end{enumerate}
  \label{def:configuration}
\end{definition}
Again, we emphasize that there are now two types of graphs, the configuration graph, and the Diameter instance, whose vertices are identified by configuration graphs.
We say configuration $H$ is a \emph{$t$-stack} configuration if $H$ has $t$ vertices. 
The vertices of our Diameter instance are identified with configurations. 
We use the following definition to specify how many nodes and vectors are in these configuration.
As we specify later, the vertices of our Diameter instance are identified by configurations of size $k$.
\begin{definition}[Size of a configuration]
  The \emph{size} of a configuration $H$ is the integer $\sum_{v\in V(H)}^{} (1+|S_v(H)|)$.
  \label{def:size}
\end{definition}
Note that the size of a configuration is the number of stacks plus the total number of vectors in all the stacks. 

\paragraph{Equivalent configurations.}
The node labels of a configuration $H$ in $[2k']$ are irrelevant except so that we can reason about operations on configurations (defined later) in a well defined way (see footnote~\ref{footnote:path}).
Accordingly, we say two configurations are \emph{equivalent} if, informally, one can be obtained by permuting the node labels of the other.
Formally, we have the following definition.
\begin{definition}[Equivalence of configurations]
\label{def:equiv}
We say two configurations $H$ and $H'$ are \emph{equivalent} if there is some permutation $\pi:[2k']\to [2k']$ such that, 
\begin{itemize}
\item Configuration $H'$ contains node $\pi(v)$ for each node $v$ of $H$, and an edge $(\pi(v),\pi(v'))$ with $(\pi(v),\pi(v'))$-edge constraint $Y^{\pi(v),\pi(v')}$ for each edge $(v,v')$ of $H$ with $(v,v')$-edge-constraint $X^{v,v'}$, such that $Y^{\pi(v),\pi(v')}_{\pi(v),j}=X^{v,v'}_{v,j}$ and $Y^{\pi(v),\pi(v')}_{\pi(v'),j}=X^{v,v'}_{v',j}$ for all $j\in[k']$, and $Y^{\pi(v),\pi(v')}_* = X^{v,v'}_*$.
\item The stacks satisfy $S_{\pi(v)}(\pi(H))=S_v(H)$ for every node $v$ of $H$.
\item The ordering $\prec_{H'}$ on $H'$  has $\pi(v)\prec_{H'} \pi(v')$ if and only if $v\prec_{H} v'$.
\item The root $\rho(\pi(H))$ of $\pi(H)$ satisfies $\rho(\pi(H))=\pi(\rho(H))$.
\end{itemize}
In this case, we write $H'=\pi(H)$.
\end{definition}
It is easy to check the following fact from Definition~\ref{def:equiv}. Taking $\pi'=\pi^{-1}$ below verifies that the equivalence in Definition~\ref{def:equiv} is indeed an equivalence relation.
\begin{lemma}
  For two permutations $\pi$ and $\pi'$, we have $\pi(\pi'(H)) = (\pi\circ\pi')(H)$
\label{lem:perm-0}
\end{lemma}

\paragraph{Edge-satisfying and valid configurations.}
For a configuration to be a valid vertex of our diameter instance, the stacks of a configuration need to satisfy particular coordinate arrays in the configuration. 
We now make precise how we want the coordinate arrays to constrain the stacks.
This is the most technical definition in the construction.
\begin{definition}[Edge-satisfying and $\mathcal{X}_v(H)$]
  \label{def:edge}
  A configuration $H$ with $s\ge 1$ vertices $v_1\prec_H\cdots\prec_H v_s$ is \emph{edge-satisfying} if and only if for every $i\in[s]$, the stack $S_{v_i}(H)$ satisfies each coordinate array in the following set $\mathcal{X}_{v_i}(H)$ of coordinate arrays.
  \begin{enumerate}
  \item For every neighbor $v'$ of $v_i$, and every index $j\in[k']$, set $\mathcal{X}_{v_i}(H)$ includes the coordinate array $X^{v_i,v'}_{v_i,j}$ and $X^{v_i,v'}_*$. Note that either $v'$ or $v_i$ is the root.
  \item For every $i' > i$, set $\mathcal{X}_{v_i}(H)$ includes the coordinate array $X^{v_{i'},v_1}_{v_{i'},i}$, where recall that $v_1$ is the root $\rho(H)$. See Figure \ref{fig:edgesatisfying}. 
  \end{enumerate}
\end{definition}

\begin{figure}
    \centering
    \includegraphics{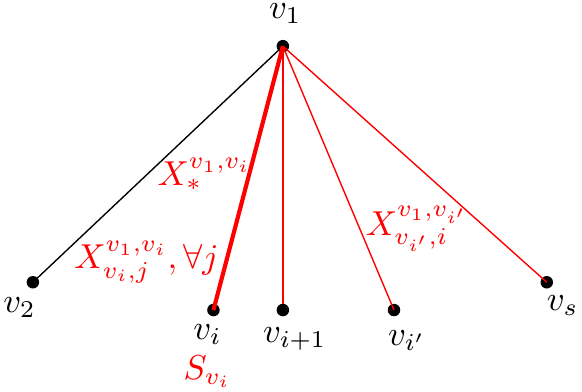}
    \caption{The coordinate arrays $\mathcal{X}_{v_i}(H)$ that stack $S_{v_i}(H)$ satisfies, for $i\ge 2$. The relevant edges are colored red, and the coordinate array that is satisfied by $S_{v_i}(H)$ is written on them. the edge $v_1v_i$ is shown in bold since many coordinate arrays on this edge-constraint are satisfied by $S_{v_i}(H)$.
    }
    \label{fig:edgesatisfying}
\end{figure}

We highlight the subtle detail that the edge constraint $X^{v_1,v_i}$ belonging to the edge $v_1v_i$ where $v_1=\rho(H)$ might hold coordinate arrays constraining the stacks of the nodes other than its endpoints $v_1$ and $v_i$. 
To get more intuition, the coordinate arrays a given stack $S_{v_i}$ needs to satisfy are illustrated in Figure~\ref{fig:edgesatisfying}, and for an edge $v_1v_i$ in configuration $H$ the stacks that must satisfy each coordinate array in $X^{v_1,v_i}$ are illustrated in Table~\ref{tab:edgesatisfying}. 
Table~\ref{tab:edgesatisfying} shows that every coordinate array in the edge-constraint $X$ constrains at most two stacks.
%Furthermore, some coordinate arrays only constrain one stack, corresponding to coordinate arrays that need to be ``remembered'' for the Yes case.

\begin{table}[]
    \centering
    \begin{tabular}{|c|c|c|c|c|c|c|c|c|c|}
    \hline
         & $*$& 1&$\ldots$ &$j$ & $\ldots$ & $i-1$ & $i$ & $\ldots$ & $k'$ \\ \hline
         $*$& $S_{v_1},S_{v_i}$&$-$ &$\ldots$ & $-$ & $\ldots$ & $-$ & $-$ & $\ldots$ & $-$ \\ \hline
         $v_1$ & $-$& $S_{v_1}$&$\ldots$& $S_{v_1}$ & $\ldots$ & $S_{v_1}$ & $S_{v_1}$ &$\ldots$& $S_{v_1}$ \\ \hline
         $v_i$ & $-$&$S_{v_1},S_{v_i}$ &$\ldots$ &$S_{v_j},S_{v_i}$ & $\ldots$ & $S_{v_{i-1}},S_{v_i}$ & $S_{v_i}$ &$\ldots$& $S_{v_i}$ \\ \hline
    \end{tabular}
    \caption{Edge satisfying constraints for $X^{v_1,v_i}$ in a configuration $H$. The entry in row $u$ and column $t$ represent the stacks satisfying $X^{v_1,v_i}_{u,t}$. The entry in row $*$ and column $*$ represent the stacks satisfying $X^{v_1,v_i}_*$. We drop $H$ in $S_{u}(H)$ for space constraints.}
    \label{tab:edgesatisfying}
\end{table}

\begin{definition}[Valid configuration]
\label{def:valid}
  The configuration $H$ is \emph{valid} if it is edge-satisfying and the stack of the root node  satisfies $|S_{\rho(H)}(H)|\ge (k-2)/2$. Here, $k$ is the parameter of our reduction. We use this definition even when the size of configuration $H$ is not $k$. 
\end{definition}
The choice of our global constant $k'$ is motivated by this definition: Since all valid configurations have a stack with at least $(k-2)/2$ vectors, all valid size-$k$ configurations, and hence all configurations at vertices of our Diameter instance, have at most $k-\ceil{(k-2)/2} = k'$ nodes. 
%Now we relate edge-constraints with the nodes of a configuration.

\paragraph{Operations on configurations.}
As mentioned earlier, our final construction consists of configurations. To relate different configurations to each other, we define operations as follows.

\begin{definition}[Operations on configurations]
  We define the following \emph{half-operations} on configurations $H$, that produce a resulting configuration $H'$.
  \begin{enumerate}
  \item \textbf{Vector insertion}. $H'$ has the same nodes, root node, edges, stacks, and ordering as $H$, except that $S_v(H') = S_v(H) + b$ for some vector $b\in A$ and some node $v$.
  \item \textbf{Vector deletion}. $H'$ has the same nodes, root node, edges, stacks, and ordering as $H$, except that $S_v(H') = popped(S_v(H))$ for some node $v$.
  \item \textbf{Node insertion}. $H'$ has the same nodes, root node, edges, stacks as $H$, except that $H'$ also contains a node $v$ labeled in $[2k']\setminus V(H)$, assigned with an empty stack $S_v(H')=\emptyset$, and an edge from node $v$ to the root node $\rho(H')=\rho(H)$ with a $(v,\rho(H'))$-edge constraint $X$, and such that the total order $\prec_{H'}$ is a total order consistent with $\prec_H$ on the nodes of $H$ and the new node $v$ as either the largest or second largest node of $\prec_H$.\footnote{This requirement that the new node $v$ is either the largest or second largest node of $\prec_H$ is not necessary, but makes the rest of the proof, especially Lemma~\ref{lem:op-4}, easier to write. Similarly, for node deletions, the deleted node does not need to be the largest or second-largest node of $\prec_H$.}
  \item \textbf{Node deletion}. $H'$ has the same nodes, root node, edges, stacks as $H$, except that for some non-root (leaf) node $v$ with $S_v(H)=\emptyset$ that is either the second-largest or largest node of $\prec_H$, $H'$ does not contain node $v$ or the edge incident to it, and the order $\prec_{H'}$ is the order $\prec_H$ restricted to the nodes of $H'$
  \item \textbf{Flip}. This half-operation is ``only" defined when $H$ has \emph{exactly two} nodes $v$ and $v'$ with $v=\rho(H)$ as the root and $|S_v(H)|=|S_{v'}(H)|$. Then $H'$ has the same nodes, edges, and stacks as $H$, but $v'=\rho(H')$ is the root of $H'$ and the ordering $\prec_{H'}$ of the nodes of $H'$ is switched accordingly, so that $v'\prec_{H'} v$.
  \end{enumerate}
  Call such a half-operation \emph{valid} if configurations $H$ and $H'$ are both valid.  
    
  A \emph{full operation} is obtained by applying a vector or node insertion, possibly applying a flip (if applicable), and then applying a vector or node deletion.
  We say a full operation from $H$ to $H'$ is \emph{valid} if each of the two (or three, if there is a flip) participating half-operations is valid, and if at least one of $H$ or $H'$ has at least two nodes.
  \label{def:op}
\end{definition}

By combining a ``delete" (node or vector) half operation to an insert (node or vector) half operation, we make sure that the endpoints of a full-operation have the same size. For examples of full operations, see Figure \ref{fig:op} and Figure \ref{fig:conf-ex2}. Full operations have the following useful properties.

%\begin{figure}
%    \centering
%    \includegraphics{conf-ex1.pdf}
%    \caption{Example of a full operation consisting of a node insertion ($v_4$) and a vector deletion (from $v_2$). We assume that $k=7$ in this example. The root in all three configurations is colored purple.}
%    \label{fig:conf-ex1}
%\end{figure}

\begin{figure}
    \centering
    \includegraphics{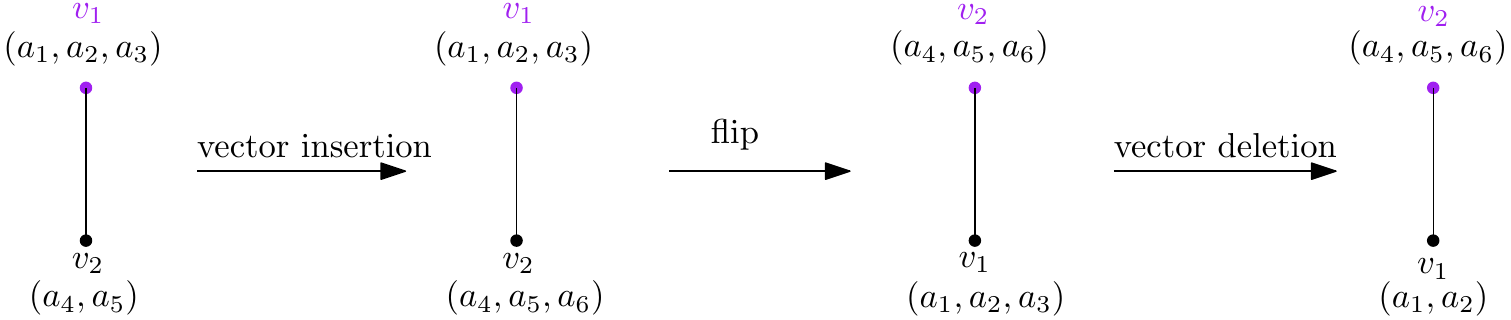}
    \caption{Example of a full operation consisting of a vector insertion (in $ v_2$), a flip and a vector deletion (from $v_1$). We assume that $k=7$ in this example. Note that when the flip operation happens, the two nodes have the same number of vectors in their stacks. The root in all four configurations is colored purple.}
    \label{fig:conf-ex2}
\end{figure}

\begin{lemma}[Properties of half and full operations]
  Let $H$ and $H'$ be configurations.
  \begin{itemize}
  \item If applying a vector insertion to $H$ gives $H'$, it is possible to apply a vector deletion to $H'$ to get $H$.
  \item If applying a vector deletion to $H$ gives $H'$, it is possible to apply a vector insertion to $H'$ to get $H$.
  \item If applying a node insertion to $H$ gives $H'$, it is possible to apply node deletion to $H'$ to get $H$.
  \item If applying a node deletion to $H$ gives $H'$, it is possible to apply node insertion to $H'$ to get $H$.
  \item Applying two flip operations to a 2-stack configuration gives the same configuration.
  \item If applying a valid full operation to $H$ gives $H'$, it is possible to apply a valid full operation to $H'$ to get $H$. 
  \end{itemize}
  \label{lem:op-1}
\end{lemma}
\begin{proof}
  For the first item, if $H'$ is obtained from a vector insertion at node $v$ in $H$, then $H$ is obtained by a vector deletion at node $v$ in $H$.
  The second, third, and fourth items are similar.
  For the fifth item, flip operations do not change the 2-node graph, and two flips preserve the root node and the ordering of the two nodes.

  For the sixth item, note that the first five items imply that every half-operation has an inverse.
  If $H'$ is obtained by applying two half-operations to $H$ that give $H''$ then $H'$, and both half operations are valid, then configurations $H,H'',H'$ are all valid configurations. 
  Then the full operation $H'\to H''\to H$ is a valid full operation.
  Similarly if $H'$ is obtained with a valid full operation including a flip, having intermediate configurations $H\to H''\to H'''\to H'$, then all the intermediate configurations are valid, and $H'\to H'''\to H''\to H$ is a valid full operation.
\end{proof}

\subsection{Defining the Diameter graph $G$}
We are now ready to define our graph $G$.
The vertex set of $G$ is the set of valid size-$k$ configurations. Recall that for all size-$k$ configurations, the number of stacks plus the total number of vectors in all stacks is $k$, and a configuration is valid if it is edge-satisfying (Definition~\ref{def:edge}) and the root stack has at least $(k-2)/2$ vectors in it. 
The edge-set of $G$ includes the following types of edges:
\begin{itemize}
\item edges $(H,H')$ such that configuration $H$ can be obtained from configuration $H'$ by a valid full operation. 
We call these edges \emph{operation edges}.
By the last part of Lemma~\ref{lem:op-1}, $(H,H')$ are connected by an operation edge \emph{if and only if} $H$ can be obtained from $H'$ by a valid full operation, so these edges can indeed by undirected.

\item weight-0 edges $(H,\pi(H))$ for all valid size-$k$ configurations $H$ and all permutations $\pi:[2k']\to [2k']$ (recall $\pi(H)$ is defined in Definition~\ref{def:equiv}).
We call these edges \emph{permutation edges}.

\item (if $k$ is even) weight-0 edges $(H,H')$ if $H'$ can be obtained by applying a flip to $H$. 
We call these edges \emph{flip edges}.
\end{itemize}
For disambiguation, we always refer to vertices of configurations as \emph{nodes}, and vertices of the Diameter instance $G$ as \emph{vertices} or \emph{configurations}.

\paragraph{Runtime analysis.}
We first show that the graph $G$ can be constructed in time $O_k(n_{OV}^{k-1}\mathbb{d}^{O(k^2)})$.
One can construct the vertices of $G$ by enumerating over all possible star graphs labeled by $[2k']$, of which there are at most $O_k(1)$, and then enumerating over all possible orderings $\prec$ of the nodes of star graphs, of which there are at most $O_k(1)$, and then enumerating over all possible stacks for each star graph, of which there are at most $O_k(n_{OV}^{k-1})$ (each configuration is size-$k$, meaning the total number of nodes (stacks) plus the total number of vectors equals $k$, and since there is always at least one node (stack), the total number of vectors is at most $k-1$), and enumerating over all possible edge-constraints, of which there are at most $O_k(\mathbb{d}^{(k'-1)\cdot (2k'+1)}) \le O_k(\mathbb{d}^{2k^2})$.
Hence, there are at most $O_k(n_{OV}^{k-1}\mathbb{d}^{2k^2})$ vertices of $G$.
Furthermore, for $t\ge 2$, there are at most $O_k(n_{OV}^{k-2}\mathbb{d}^{2k^2})$ many $t$-stack configurations of $G$.

For any configuration, there are $O_k(n_{OV})$ vector insertions possible, $O_k(1)$ vector deletions possible, $O_k(\mathbb{d}^{2k'+1})$ node insertions possible, and $O_k(1)$ node deletions possible.
Hence, each configuration of $G$ has at most $O_k(n_{OV}+\mathbb{d}^{2k'+1})$ neighbors.
Every edge of $G$ has at least one endpoint that has $t\ge 2$ stacks (by definition of valid full operation), so the total number of edges of $G$ is at most $O_k(n_{OV}+\mathbb{d}^{2k'+1})\cdot O_k(n_{OV}^{k-2}\mathbb{d}^{2k^2})\le O_k(n_{OV}^{k-1}\mathbb{d}^{4k^2})$.

Hence, $G$ has $\tilde O(n_{OV}^{k-1})$ vertices (configurations) and edges.
Checking whether any half-operation is valid takes time $O_k(\mathbb{d}) = \tilde O_k(1)$. 
Hence enumeration of vertices (configurations) and edges of the Diameter graph $G$ is standard and can be done in time near-linear in the number of vertices and edges, so the construction of $G$ takes time $\tilde O(n_{OV}^{k-1})$.

\subsection{Some useful properties of configurations}

We now move on to proving the correctness of our configurations, showing that the Diameter is at least $2k-1$ when the $k$-OV instance $\Phi$ has a solution (Yes case), and the Diameter is at most $k$ when $\Phi$ has no solution (No case).
We begin with some useful lemmas about configurations.

\paragraph{Lemma for the No case.}
In the No case, we need to construct length $k$ paths between every pair of nodes and verify that those paths are valid paths in the Diameter instance.
The following natural lemma facilitates these verifications.
Call $H'$ a  \emph{subconfiguration} of $H$ if $H'$ can be obtained from $H$ by vector deletions and node deletions.
\begin{lemma}
  If $H'$ is a subconfiguration of $H$ and $H$ is edge-satisfying, then $H'$ is also edge-satisfying.
  \label{lem:op-4}
\end{lemma}
\begin{proof}
  It suffices to prove that if $H'$ is obtained by applying a single vector deletion or node deletion to $H$, and $H$ is valid, then $H'$ is valid.
  The full lemma follows from induction of the number of deletions needed to obtain $H'$ from $H$.
  Let $H$ have vertices $v_1\prec_H\cdots\prec_H v_s$.

  Suppose $H'$ is obtained from $H$ by a vector deletion.
  Then $H$ and $H'$ have the same node set and edge set.
  Let $i\in[s]$.
  In the Definition~\ref{def:edge}, the set of coordinate arrays $\mathcal{X}_{v_i}(H')$ is the same as the set of coordinate arrays $\mathcal{X}_{v_i}(H)$, because $H$ and $H'$ are the same graph with the same edge-constraints.
  Since we assume $H$ is edge-satisfying, we have that $S_{v_i}(H)$ satisfies all the coordinate arrays in $\mathcal{X}_{v_i}(H)=\mathcal{X}_{v_i}(H')$, so $S_{v_i}(H')$ does as well, by Lemma~\ref{lem:stack-1}.
  This holds for all $i\in[s]$, so we have that $H'$ is edge-satisfying.

  Now suppose $H'$ is obtained from $H$ by a node deletion, so that the graph $H'$ is a subgraph of the graph $H$ with a leaf node deleted.
  We claim that, for all $i$ such that node $v_i$ is in $H'$, we have $\mathcal{X}_{v_i}(H')\subseteq\mathcal{X}_{v_i}(H)$.
  First, suppose $v_s$ is deleted from $H$ to give $H'$.
  Then, for each $i=1,\dots,s-1$, by Definition~\ref{def:edge}, the set $\mathcal{X}_{v_i}(H')$ is the same as the set of coordinate arrays $\mathcal{X}_{v_i}(H)$, except with $X^{v_s,v_1}_{v_s,i}$ deleted, and, if $v_i$ is a neighbor of $v_s$ (only true for $i=1$), with coordinate arrays $X^{v_s,v_i}_{v_i, j}$ deleted for $j\in[k']$, so indeed $\mathcal{X}_{v_i}(H')\subset\mathcal{X}_{v_i}(H)$.
  Now suppose $v_{s-1}$ is deleted from $H$ to give $H'$.
  For $1\le i\le s-2$, we have $\mathcal{X}_{v_i}(H')\subset\mathcal{X}_{v_i}(H)$ by the same reasoning as when $v_s$ is deleted.
  Additionally, we can show $\mathcal{X}_{v_s}(H')=\mathcal{X}_{v_s}(H)$: nodes $v_s$ and $v_{s-1}$ are not adjacent in $H$ (node deletions can only delete non-root nodes so $v_{s-1}$ is not the root) so all of the coordinate arrays of $\mathcal{X}_{v_s}(H)$ and $\mathcal{X}_{v_s}(H')$ in part 1 of Definition~\ref{def:edge} are the same, and $\mathcal{X}_{v_s}(H)$ and $\mathcal{X}_{v_s}(H')$ have no coordinate arrays in part 2 of Definition~\ref{def:edge} since $v_s$ is the largest node each of $\prec_{H}$ and $\prec_{H'}$.

  Thus, we have that $\mathcal{X}_{v_i}(H')\subseteq\mathcal{X}_{v_i}(H)$ for all nodes $v_i$ in $H'$.
  For all nodes $v_i$ in $H'$, we have the stacks $S_{v_i}(H')$ and $S_{v_i}(H)$ are the same, since no vector insertions/deletions were applied.
  Thus, since stack $S_{v_i}(H)$ satisfies all the coordinate arrays in $\mathcal{X}_{v_i}(H)$, we have $S_{v_i}(H')$ satisfies all the coordinate arrays in $\mathcal{X}_{v_i}(H')$, as desired.

  We have shown that if $H'$ is obtained by applying a single vector deletion or node deletion to $H$, and $H$ is valid, then $H'$ is valid.
  By the first paragraph of the proof, this completes the proof.
\end{proof}

%\begin{lemma}
%  Suppose $k$ is odd.
%  Let $H$ and $H'$ be configurations such that $H$ is obtained from applying a valid full operation consisting of a vector insertion at some node $v$, followed by a flip, followed by a vector deletion at some node $v'\neq v$. Then the full operation is valid if and only if both $H$ and $H'$ are valid.
%\end{lemma}
%\begin{proof}
%\end{proof}

\paragraph{Lemma for the Yes case.}
The next lemma is useful for the Yes case.
\begin{figure}
    \centering
    \includegraphics{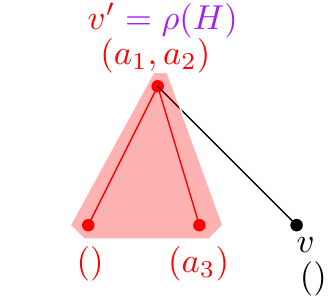}
    \caption{Lemma~\ref{lem:yes-3}. In the example configuration of size $k=7$, to delete the root node $v'=\rho(H)$ (purple) without deleting $v$, one needs to delete all the red vectors and red nodes. This requires 3 node deletions and 3 vectors deletions for a total of $6=k-1$ deletions.}
    \label{fig:yes-lemma}
\end{figure}
\begin{lemma}
  Suppose $H$ is a size-$k$ configuration containing a non-root leaf node $v$ with $S_v(H)=\emptyset$ and edge $(v,v')$ where $v'=\rho(H)$.
  Suppose that one applies $c$ full operations among which node $v'$ is deleted but node $v$ is never deleted.
  Then $c\ge k-1$.
\label{lem:yes-3}
\end{lemma}
\begin{proof}
  Let $H_0=H,H_1,\dots,H_c$ be the sequence of configurations such that $H_i$ is the result of applying a valid full operation to $H_{i-1}$ for $i=1,\dots,c$.
  Let $v''\notin \{v,v'\}$ be an arbitrary node in $H$.
  We claim that node $v''$ must be deleted before node $v'$.
  Let $i\in\{0,\dots,c\}$ be the largest index such that $v''$ and $v$ are both in configuration $H_i$.
  Node $v'$ is on the path from node $v''$ to node $v$ in configuration graph $H_0$.
  Only leaf nodes can be deleted in a node deletion.
  Thus, as $v$ and $v''$ are both in $H_0,\dots,H_i$, all the nodes on the path from $v$ to $v''$ are also nodes in $H_0,\dots,H_i$.
  In particular, node $v'$ is in $H_0,\dots,H_i$, so node $v''$ must be deleted before $v'$. 
  
  Hence, the only way to delete node $v'$ without deleting node $v$ is to first delete all nodes other than $v'$ (by first deleting the vectors in their stacks and then the node) and then deleting $v'$.
  This results in deleting all nodes other than $v$, which takes at least $\sum_{u\in V(H)}^{} (1+|S_u(H)|) - (1+|S_v(H)|) = k-(1+0)$ deletions.
  Since each full operation applies at most one deletion, the number of full operations needed to delete $v'$ without deleting $v$ is at least $k-1$.
\end{proof}

\paragraph{Permutations commute with valid full operations}

The next few lemmas justify the informal statement that ``permutations commute with valid full operations''. 
This statement is convenient in the Yes case because it allows us to assume that all permutation edges are at the end of a path.
Intuitively, we expect this lemma to be true because changing the node labels of a configuration gives essentially the same configuration.
\begin{lemma}
  \label{lem:perm-1}
  Let $\pi:[2k']\to [2k']$ be a permutation.
  Let $H$ be a configuration, and suppose that applying a vector insertion (vector deletion, node insertion, node deletion, flip) on $H$ gives configuration $H'$.
  Then there exists a vector insertion (vector deletion, node insertion, node deletion, flip) that, applied to $\pi(H)$, gives $\pi(H')$.
\end{lemma}
\begin{proof}
  A vector $b\in A$ is inserted at node $v$ in $H$ ($v$ is a node label in $[2k']$) to give a configuration $H'$. 
  Suppose that inserting vector $b$ at node $\pi(v)$ in $\pi(H)$ gives a configuration $H''$.
  We claim $H''=\pi(H')$.
  By definition of vector insertion, $H''$ has the same nodes, edges, edge-constraints, root node, and ordering as configuration $\pi(H)$.
  Furthermore, since $H$ has the same nodes, edges, edge-constraints, root node, and ordering as configuration $H'$, we have $\pi(H)$ and $\pi(H')$, and thus $H''$ and $\pi(H')$ have the nodes, edges, edge-constraints, root node, and ordering.
  Furthermore, the stacks $S_{\pi(v)}(H'')$ and $S_{\pi(v)}(\pi(H'))$ are both equal to $S_v(H) + b$, and the stacks $S_{\pi(v')}(H'')$ and $S_{\pi(v')}(\pi(H'))$ are both equal to $S_{v'}(H)$ for nodes $v'\neq v$ in $H$, so we indeed have $H''=\pi(H')$.

  This proves the lemma for vector insertions.
  The proofs for vector deletions, node insertions, node deletions, and flips are similar.
\end{proof}
\begin{lemma}
  Let $\pi:[2k']\to [2k']$ be a permutation.
  If configuration $H$ is valid, then configuration $\pi(H)$ is valid.
\label{lem:perm-2}
\end{lemma}
\begin{proof}
  The root node $\rho(\pi(H))$ of $\pi(H)$ has the same stack as the root node $\rho(H)$ of $H$, which has at least $(k-2)/2$ vectors. 
  By definition of $\pi(H)$, for each node $v\in V(H)$, the set of coordinate arrays $\mathcal{X}_{\pi(v)}(H)$ is the same as $\mathcal{X}_{v}(H)$.
  Since $H$ is valid, $S_v(H)$ satisfies every coordinate array in $\mathcal{X}_v(H)$, so $S_{\pi(v)}(\pi(H))=S_v(H)$ satisfies every coordinate array in $\mathcal{X}_{\pi(v)}(\pi(H))=\mathcal{X}_v(H)$.
  This holds for all $v$, so $\pi(H)$ is edge-satisfying and thus valid.
\end{proof}
As a corollary of Lemmas~\ref{lem:perm-1} and \ref{lem:perm-2}, we have that permutations commute with valid full operations.
\begin{corollary}
  Let $\pi:[2k']\to [2k']$ be a permutation.
  Let $H$ be a configuration, and suppose that applying some valid full operation on $H$ gives configuration $H'$.
  Then applying some valid full operation on $\pi(H)$ gives $\pi(H')$.
  \label{cor:perm-3}
\end{corollary}

\subsection{No case.}
We now prove that when $\Phi$ has no solution, our Diameter instance has diameter at most $k$.
To do so, we find a length $k$ path between any two configurations $H$ and $H'$.
As sketched in the overview, we apply $k$ full operations to get $H'$ from $H$, and each operation inserts a vector or node ``from $H'$'' and deletes a vector or node ``from $H$''.
For an example of such a path when $k=7$, see Figure~\ref{fig:no-case}.

\begin{figure}
    \centering
    \includegraphics{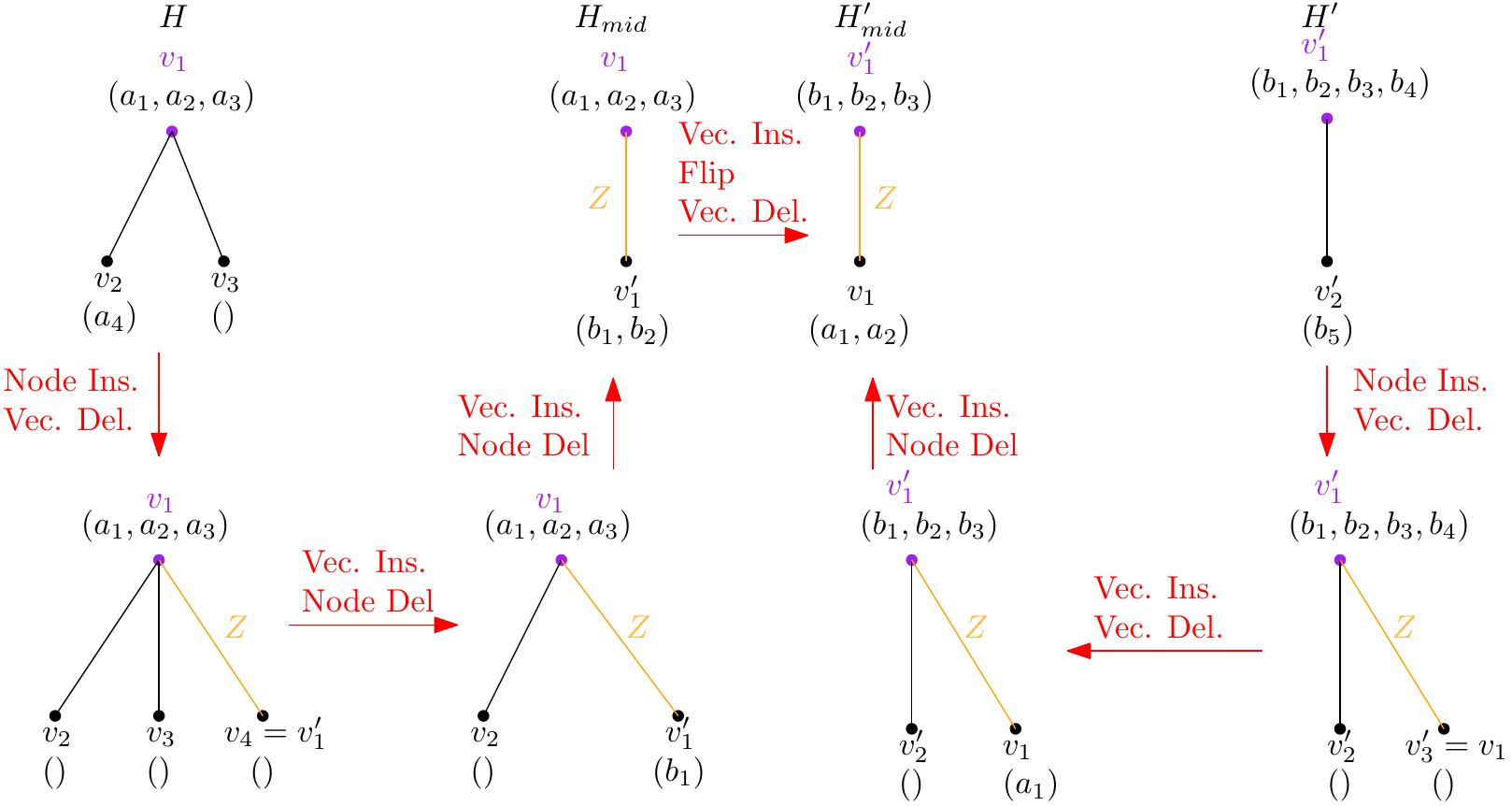}
    \caption{The path of length $7$ between $H$ and $H'$ for $k=7$. Full operations are indicated by red arrows and roots are indicated by purple.
    The ``extra'' edge-constraint $Z$ that belongs to neither $H$ nor $H'$ is labeled in orange.  }
    \label{fig:no-case}
\end{figure}

Let $H$ be an arbitrary size-$k$ configuration with vertices $v_1\prec_H\cdots\prec_H v_s$ for some $s\ge 1$, where $v_1=\rho(H)$ is the root, and with edges $v_1v_i$ with edge-constraint $X^{v_1,v_i}$ for $2\le i\le s$.
Let $H'$ be an arbitrary size-$k$ configuration with vertices $v'_1\prec_{H'}\dots\prec_{H'}v'_{s'}$ for some $s'\ge 1$, where $v'_1=\rho(H')$ is the root, and with edges $v'_1v'_i$ with edge-constraint $Y^{v_1',v_i'}$ for $2\le i\le s$.
By taking a permutation edge (of weight 0) from vertex $H'$ in the Diameter instance $G$ to obtain an equivalent configuration, we may assume without loss of generality that the set of node labels $\{v_1,\dots,v_s\}$ of $H$ are disjoint from the node labels $\{v_1',\dots,v_{s'}'\}$ of $H'$.

We now define an edge-constraint $Z$, containing the only ``extra'' coordinate arrays we need in the path from $H$ to $H'$.
Let $Z$ be a $(v_1,v_1')$-edge constraint such that, 
\begin{itemize}
\item For $i\in[k']$, coordinate array $Z_{v_1,i}$ is satisfied by stack $S_{v_1}(H)$ and, if $i\le s'$, by stack $S_{v'_i}(H')$, 
\item For $i\in[k']$, coordinate array $Z_{v'_1,i}$ is satisfied by stack $S_{v'_1}(H')$ and, if $i\le s$, by stack $S_{v_i}(H)$, and
\item $Z_{*}$ is satisfied by $S_{v_1}(H)$ and $S_{v_1'}(H')$.
\end{itemize}
As configurations $H$ and $H'$ are size-$k$ and have at least 1 stack, any stack of $H$ or $H'$ has at most $k-1$ vectors.
Hence, the coordinate arrays of $Z$ all exist by Lemma~\ref{lem:stack-2}.
Note that the definition of $Z$ is symmetric with respect to $H$ and $H'$, in the sense that if we switch $H$ with $H'$ (and $s$ with $s'$ and $(v_1,\dots,v_s)$ with $(v_1',\dots,v_s')$), the definition of $Z$ stays the same.

We now define two intermediate nodes $H_{mid}$ and $H_{mid}'$, which are on our desired path from $H$ to $H'$.
Let $H_{mid}$ be the configuration with nodes $v_1$ and $v'_1$, with the connecting edge having $(v_1,v_1')$-edge constraint $Z$, where 
\begin{itemize}
    \item $v_1=\rho(H_{mid})$ is the root, 
    \item $S_{v_1}(H_{mid})$ is the bottom $\ceil{(k-2)/2}$ elements of $S_{v_1}(H)$, and 
    \item $S_{v'_1}(H_{mid})$ is the bottom $\floor{(k-2)/2}$ elements of $S_{v'_1}(H')$. 
\end{itemize}

Let $H_{mid}'$ be the configuration with nodes $v_1$ and $v'_1$, with the connecting edge having $(v_1,v_1')$-edge constraint $Z$, where 
\begin{itemize}
    \item $v'_1=\rho(H_{mid}')$ is the root, 
    \item $S_{v'_1}(H_{mid}')$ is the bottom $\ceil{(k-2)/2}$ elements of $S_{v'_1}(H')$, and 
    \item $S_{v_1}(H_{mid}')$ is the bottom $\floor{(k-2)/2}$ elements of $S_{v_1}(H)$. 
\end{itemize}

We have that $H_{mid}$ and $H_{mid}'$ are valid: by the definition of the edge-constraint $Z$, we have that $S_{v_1}(H)$ and thus $S_{v_1}(H_{mid})$ satisfies $Z_{v_1,j}$ for all $j\in[k']$, and also satisfies coordinate array $Z_{v_1',1}$ and $Z_{*}$.
Similarly, $S_{v'_1}(H')$ and thus $S_{v'_1}(H_{mid})$ satisfies $Z_{v_1',j}$ for all $j\in[k']$, and also satisfies coordinate arrays $Z_{*}$.
Thus, $H_{mid}$ is edge-satisfying and thus valid.
By a symmetric argument, $H_{mid}'$ is also valid.
Note that $H_{mid}$ and $H_{mid}'$ are symmetric with respect to $H$ and $H'$, in the sense that if we switched $H$ and $H'$, then $H_{mid}$ becomes $H_{mid}'$ and vise-versa.

\begin{figure}
    \centering
    \includegraphics{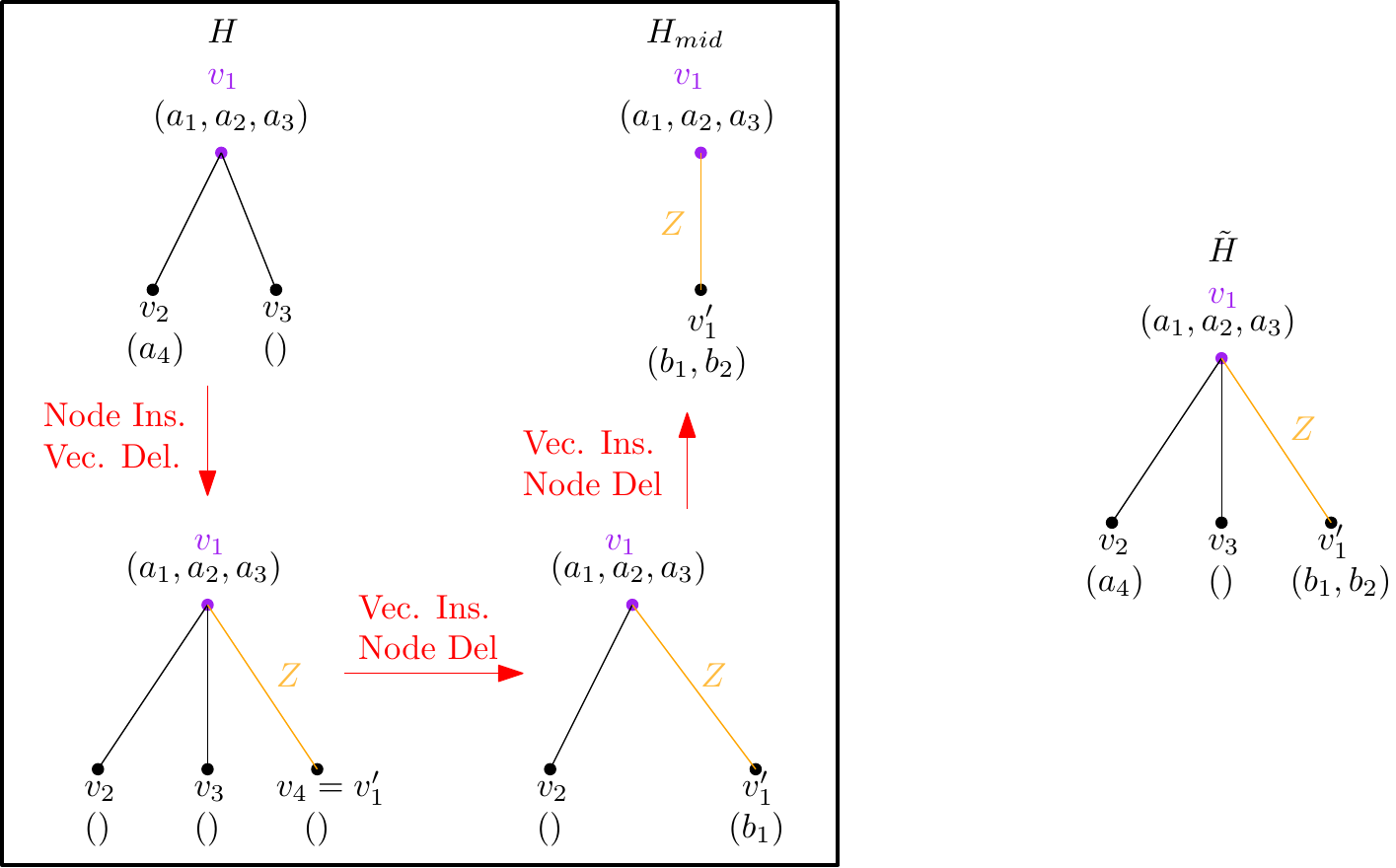}
    \caption{Claim~\ref{lem:no-4}, the configuration $\tilde H$ for Figure~\ref{fig:no-case}: all configurations on the path from $H$ to $H_{mid}$ are subconfigurations of $\tilde H$. By Lemma~\ref{lem:op-4}, showing $\tilde H$ is valid implies that the path from $H$ to $H_{mid}$ is valid.
    }
    \label{fig:no-case-tildeH}
\end{figure}

\begin{claim}
  One can apply $\floor{k/2}$ valid full operations on $H$ to obtain $H_{mid}$, and $\floor{k/2}$ valid full operations on $H'$ to obtain $H_{mid}'$.
\label{lem:no-4}
\end{claim}
\begin{proof}
  We prove this for $H$ and $H_{mid}$, and the result for $H'$ and $H_{mid}'$ follows from a symmetric argument (the symmetry holds because the definition of $Z$ and the definitions of $H_{mid}$ and $H_{mid}'$ are symmetric with respect to $H$ and $H'$). 
  Let $\tilde H$ be the configuration obtained by adding node $v'_1$ to $H$ with stack $S_{v'_1}(H_{mid}')$ (of size $\floor{(k-2)/2}$), with an edge $(v_1,v_1')$ having edge constraint $Z$, and such that the ordering $\prec_{\tilde H}$ agrees with $\prec_H$ on the nodes of $H$, and $v'_1$ is the largest node of $\prec_{\tilde H}$ (see Figure~\ref{fig:no-case-tildeH}).
  Note that $\tilde H$ has size larger than $k$ (to be precise, it has size $k+\floor{k/2}$).

  We first prove that $\tilde H$ is edge-satisfying.
  First, the set $\mathcal{X}_{v_1'}(\tilde H)$ has coordinate arrays $Z_*$ and $Z_{v_1',j}$ for $j\in[k']$, by part 1 of Definition~\ref{def:edge}, and has no coordinate arrays from part 2 of Definition~\ref{def:edge} as $v'_1$ is the largest node of $\prec_H$.
  By definition of $Z$, stack $S_{v'_1}(H')$ satisfies all these coordinate arrays, and thus by Lemma~\ref{lem:stack-2} stack $S_{v'_1}(H_{mid})$ does as well,   satisfying the requirement of Definition~\ref{def:edge} for node $v_1'$.
  For $i\in [s]$, the set of coordinate arrays in $\mathcal{X}_{v_i}(\tilde H)$ is the same as the set of coordinate arrays $\mathcal{X}_{v_i}(H)$ plus the coordinate array $Z_{v_1',i}$, and, if $i=1$, plus the coordinate arrays $Z_*$ and $Z^{v_1', v_1}_{v_1,j}$ for $j\in[k']$.
  By definition of $Z$, we have that $S_{v_i}(\tilde H) = S_{v_i}(H)$ satisfies coordinate array $Z_{v_1',i}$.
  Furthermore, $S_{v_1}(\tilde H)=S_{v_1}(H)$ satisfies coordinate arrays $Z_*$ and $Z_{v_1,j}$ for $j\in[k']$.
  Since configuration $H$ is edge-satisfying and the above coordinate arrays are satisfied, we conclude that configuration $\tilde H$ is edge-satisfying.

  We now note that $H_{mid}$ can be obtained from $H$ by applying the following half-operations
  \begin{itemize}
  \item Insert node $v_1'$ as the largest node in the ordering
  \item Insert vectors into $v_1'$ $\floor{(k-2)/2}$ times.
  \item For each $i=s,s-1,\dots,2$, delete vectors from $S_{v_i}$ until the stack is empty, and then delete node $v_i$.
  \item Delete vectors from $S_{v_1}$ until the stack has size $\ceil{(k-2)/2}$.
  \end{itemize}
  We can check that there are $\floor{k/2}$ insertions and $\sum_{i=1}^{s} (1+|S_{v_i}|) - (1+\ceil{k-2}/2) = k-\ceil{k/2}=\floor{k/2}$ deletions.
  We can obtain $H$ from $H_{mid}$ by alternating applying these insertions and deletions, giving a configurations $H=H_0,H_{0.5},H_1,\dots,H_{\floor{k/2}-0.5},H_{\floor{k/2}}=H_{mid}$, so that applying the $i$th insertion to $H_{i-1}$ gives the size-$k+1$ configuration $H_{i-0.5}$, and applying the $i$th deletion to $H_{i-0.5}$ gives the size-$k$ configuration $H_i$.
  These half-operations indeed satisfy the definition of half-operations: all the vector insertions/deletions are legal, the single node insertion is legal as $v_1'$ is inserted as the largest node, and all the node deletions are legal as the deleted nodes are always the second-largest node in the ordering.
  Furthermore, if we perform only the insertions, we obtain configuration $\tilde H$.
  Hence, any $i=0,0.5,\dots,\floor{k/2}$, we can obtain configuration $H_i$ from configuration $\tilde H$ by applying vector deletions at node $v_1'$ until stack $S_{v_1'}$ is the right size, and then applying the first $\floor{i}$ node/vector deletions above (at nodes $v_s,v_{s-1},\dots$).
  Thus, for $i=0,0.5,\dots,\floor{k/2}$ configuration $H_i$ is a subconfiguration of configuration $\tilde H$. 
  Since configuration $\tilde H$ is valid, by Lemma~\ref{lem:op-4}, each $H_i$ and $H_{i+0.5}$ is valid, so we have a sequence of $\floor{k/2}$ valid full operations that gives $H_{mid}$ from $H$.
\end{proof}

With Claim~\ref{lem:no-4}, we have nearly proved the No case.
It remains to show that $H_{mid}$ and $H_{mid}'$ are at distance either 0 or 1, depending on the parity of $k$.

If $k$ is even, then $H_{mid}$ can be obtained by applying a flip to $H_{mid}'$, and thus the two configurations are at distance 0 in the Diameter graph $G$.
Thus, there is a length $2\cdot \floor{k/2} = k$ path from $H$ to $H'$ through $H_{mid}$ and $H_{mid}'$ by Claim~\ref{lem:no-4}.

If $k$ is odd, then $H_{mid}$ is distance 1 from $H_{mid}'$:
$H_{mid}'$ is obtained from $H_{mid}$ by applying a vector insertion at node $v_1'$, giving a configuration $H_{mid,+}$, followed by a flip, giving a configuration $H_{mid,+}'$, followed by a vector deletion at node $v_1$, giving configuration $H_{mid}$.
The flip can be done because $H_{mid,+}$ and $H_{mid,+}'$ both have two nodes, each of which has a stack of size $\ceil{(k-2)/2}$.
We now check these half-operations are all valid operations, by checking that configurations $H_{mid,+}$ and $H_{mid,+}'$ are valid configurations.
Since no vectors are deleted at node $v_1'$ from $H_{mid,+}$ to $H_{mid}'$, we have $S_{v_1'}(H_{mid,+}) = S_{v_1'}(H_{mid}')$ is a substack of $S_{v_1'}(H')$, and similarly $S_{v_1}(H_{mid,+})=S_{v_1}(H_{mid})$ is a substack of $S_{v_1}(H)$.
Hence, by construction of $Z$ and Lemma~\ref{lem:stack-1}, stack $S_{v_1'}(H_{mid,+}) = S_{v_1'}(H_{mid}')$ satisfies coordinate array $Z_{v_1',j}$ for all $j\in[k']$ and also satisfies coordinate array $Z_*$, and the stack $S_{v_1}(H_{mid,+}) = S_{v_1}(H_{mid})$ at the root node of $H_{mid,+}$ satisfies $Z_{v_1,j}$ for all $j\in[k']$ and also satisfies coordinate arrays $Z_{v_1',1}$ and $Z_*$.
Hence, configuration $H_{mid,+}$ is edge-satisfying, and thus a valid configuration.
By a symmetric argument, configuration $H_{mid,+}'$ is valid.
Hence, configurations $H_{mid}$ and $H_{mid}'$ are adjacent in the diameter instance with an edge of weight 1, and we have a path from $H$ to $H'$ through $H_{mid}$ and $H_{mid}'$ of length $1+2\cdot \floor{k/2} = k$ by Claim~\ref{lem:no-4}. 

In either case, we have shown that, when $A$ has no $k$ orthogonal vectors, then for any two configurations $H$ and $H'$, there is a length $k$ path from $H$ to $H'$.
This completes the proof of the no case.

\subsection{Yes case.}
We now prove that the Diameter of $G$ is at least $2k-1$ in the Yes case.
Suppose $A$ has an orthogonal $k$-tuple $(a_1,\dots,a_k)$.
Throughout this section fix $v\in[2k']$ to be an arbitrary edge label (say $v=1$).
Let $H$ be the 1-stack configuration with a single node $v$ assigned with a stack $S_v(H)=(a_1,\dots,a_{k-1})$ (and a trivial ordering).
Let $H'$ be the 1-stack configuration with a single node labeled $v$ assigned with a stack $S_{v}(H')=(a_k,\dots,a_2)$. 
We claim configurations $H$ and $H'$ are at distance $2k-1$ in the Diameter graph $G$.

Consider a path $H_0=H,H_1,\dots,H_{r+1}=H'$ from $H$ to $H'$ using edges of $G$, and assume for contradiction this path has length $2k-2$ (if it has length less than $2k-2$, we may assume without loss of generality that in one of the $t$-stack vertices for $t\ge 2$, there are trivial valid full operations (e.g., node insertion followed by node deletion), which give self loop edges of weight 1, increasing the path length to $2k-2$).
This path contains some valid full operation edges, possibly some weight-0 flip edges if $k$ is even, and possibly some weight-0 permutation edges between equivalent configurations.
By Corollary~\ref{cor:perm-3}, we may assume without loss of generality that all weight-0 permutation edges are at the end of the path, and furthermore if there are multiple permutations $\pi_1,\dots,\pi_\ell:[2k']\to[2k']$, we may replace them by a single permutation $\pi=\pi_1\circ\cdots\pi_\ell$ by Lemma~\ref{lem:perm-0}.
Hence, we may assume that our path has $2k-2$ valid full operation edges, followed by a single weight-0 edge applying a permutation $\pi$.

Thus, we may assume that $r=2k-2$, and configuration $H'$ is $\pi(H_{2k-2})$ for some $\pi:[2k']\to[2k']$, so that configuration $H_{2k-2}$ contains a single stack at node $v'\defeq \pi^{-1}(v)$, and so that for $i=1,\dots,2k-2$, configuration $H_i$ can be reached from $H_{i-1}$ by an operation edge, and possibly a flip edge.
For each $i=0,\dots,2k-3$, the valid full operation on $H_i$ has one valid vector/node insertion, possibly followed by a flip operation, followed by one valid vector/node deletion, possibly followed by a flip, so we can let $H_{i+0.5}$ denote the result of only applying the insertion and possibly a flip to $H_i$, so that $H_{i+1}$ is the result of applying a deletion, possibly followed by a flip, to $H_{i+0.5}$.
By definition of a valid half-operation, configuration $H_{i+0.5}$ is valid (and has size $k+1$).

The following two claims reason about the stacks and the edge-constraints that must be on the path.
\begin{claim}
  If an edge $(w,w')$ appears in configuration $H_i$ for some integer $i=1,\dots,2k-3$, it also appears (with the corresponding edge-constraint) in configurations $H_{i-0.5}$ and $H_{i+0.5}$. 
\label{lem:yes-1.6}
\end{claim}
\begin{proof}
  Configuration $H_{i+0.5}$ is obtained by applying a vector or node insertion to $H_i$, possibly followed by a flip, so no node, and thus no edge is deleted from $H_i$ to $H_{i+0.5}$.
  Configuration $H_{i}$ is obtained by applying a vector or node deletion to $H_{i-0.5}$, possibly followed by a flip, so $H_{i-0.5}$ is obtained by possibly applying a flip to $H_i$, followed by a node or vector insertion, and again no edge is deleted.
\end{proof}

\begin{claim}
  For $0\le s\le k-1$, we have $S_v(H_s)$ and $S_v(H_{s+0.5})$ both contain $(a_1,\dots,a_{k-1-s})$ as a substack.
  For $k-1\le s\le 2k-2$, we have $S_{v'}(H_{s})$ and $S_{v'}(H_{s-0.5})$ both contain $(a_k,\dots,a_{2k-s})$ as a substack.
\label{lem:yes-2}
\end{claim}
\begin{proof}
  For the first item, we have $S_v(H_0)=(a_1,\dots,a_{k-1})$, and each of the first $s$ full operations deletes at most one vector from this stack, so stack $S_v(H_s)$ has $(a_1,\dots,a_{k-1-s})$ as a substack.
  By Claim~\ref{lem:yes-1.6}, $S_v(H_{i+0.5})$ does as well.
  Similarly, we have stack $S_v(H_{2k-2})=(a_k,\dots,a_2)$.
  Applying $2k-2-s$ full operations from $H_{2k-2}$ gives $H_{s}$, but each operation deletes at most one vector from the starting stack $S_{v'}(H_{2k-2})=(a_k,\dots,a_2)$.
  Hence, stack $S_v(H_{s})$ has $(a_k,\dots,a_{2k-s})$ as a substack, and by Claim~\ref{lem:yes-1.6}, stack $S_v(H_{s-0.5})$ does as well.
\end{proof}

Let $s$ be the largest index such that node $v$ is in configurations $H_0,\dots,H_s$ ($s$ exists because $H_0$ contains node $v$).
Let $s'$ be the smallest index such that node $v'$ is in configurations  $H_{s'},\dots,H_{2k-2}$ (again $s'$ exists because $H_{2k-2}$ contains node $v'$).
By the maximality of $s$ (and since we assume no permutation edges are used in $H_0,\dots,H_{2k-2}$), node $v$ has an empty stack in configuration graph $H_s$.
Node $v$ also has a size $k-1$ stack in $H_0$.
Since each valid full operation can delete at most one vector from some stack, we have that $s\ge k-1$.
Similarly, we have that $s'\le k-1$, so $s'\le s$.
Thus, nodes $v$ and $v'$ both appear in each of the configurations $H_{s'},\dots,H_s$.
We have three cases, and in each case, we show that our path contradicts Lemma~\ref{lem:yes-1}.

\textbf{Case 1. $v=v'$.}
This implies that $s'=0$ and $s=r$, and node $v$ appears in every configuration $H_0,\dots,H_{2k-2}$.
We have that the stack $S_v(H_0)=(a_1,\dots,a_{k-1})$, and $S_v(H_{2k-2})=(a_k,\dots,a_2)$.
Thus, to obtain $S_v(H_{r})$ from $S_v(H_0)$, one needs to apply $k-1$ vector deletions followed by $k-1$ vector insertions.
Since each valid full operation applies at most one vector insertion followed by at most one vector deletion, the first $k-1$ full operations of our path must include a vector deletion at node $v$, and the last $k-1$ edges must include a vector insertion at node $v$, inserting the vectors $a_k,\dots,a_2$ in that order.
In particular, we have $S_v(H_k) = (a_k)$.

Because valid full operations must have one endpoint with at least two nodes, $H_0$ to $H_1$ operation must include a node insertion of some node $w\neq v$ with an edge $(v,w)$.
By Claim~\ref{lem:yes-1.6} the edge $(v,w)$ with an edge constraint $X^{v,w}$ appears in configuration $H_{0.5}$, so stack $S_v(H_{0.5}) = (a_1,\dots,a_{k-1})$ satisfies coordinate array $X^{v,w}_{*}$.
Furthermore, since there are no node-deletions in the first $k-1$ valid full operations (because each full operation deletes either a vector or node, not both), we know the edge $(v,w)$ exists in each of $H_1,\dots,H_{k-1}$. 
By Claim~\ref{lem:yes-1.6} the edge $(v,w)$ labeled with the edge constraint $X^{v,w}$ also exist in configuration $H_{k-0.5}$.
Additionally, as we reasoned earlier, $S_v(H_{k-0.5})=(a_k)$, so stack $(a_k)$ satisfies coordinate array $X^{v,w}_*$.
However, this means that stacks $(a_1,\dots,a_{k-1})$ and $(a_k)$ both satisfy $X^{v,w}_*$, which is a contradiction of Lemma~\ref{lem:yes-1}.

\textbf{Case 2. $v\neq v'$ and nodes $v$ and $v'$ are adjacent in configuration $H_{s'}$.}
Clearly we have $s'\ge 1$ and $s\le 2k-3$ in this case.
In configuration $H_{s'}$, node $v'$ is a non-root leaf node with an empty stack $S_{v'}(H_{s'})=\emptyset$ and incident edge $(v,v')$.
Furthermore, from configuration $H_{s'}$ to $H_{s+1}$, node $v$ is deleted, but node $v'$ is in configurations $H_{s'},\dots,H_{s+1}$.
Hence, by Lemma~\ref{lem:yes-3}, we have $(s+1)-s'\ge k-1$.

By Claim~\ref{lem:yes-1.6}, both configurations $H_{s'-0.5}$ and $H_{s+0.5}$ contain the edge $(v,v')$ with edge-constraint $X^{v,v'}$.
By Claim~\ref{lem:yes-2}, in configuration $H_{s'-0.5}$, node $v$ is labeled with a stack that contains $(a_1,\dots,a_{k-s'})$ as a substack, so by Lemma~\ref{lem:stack-1}, stack $(a_1,\dots,a_{k-s'})$ satisfies coordinate array $X^{v,v'}_*$.
Similarly, by Claim~\ref{lem:yes-2}, in configuration $H_{s+0.5}$, node $v'$ is labeled with a stack that contains $(a_k,\dots,a_{2k-1-s})$ as a substack, so stack $(a_k,\dots,a_{2k-1-s})$ satisfies coordinate array $X^{v,v'}_*$.
Since $k-s'\ge (2k-1-s)-1$, we have that, for $j=k-s'$, both stacks $(a_1,\dots,a_j)$ and $(a_k,\dots,a_{j+1})$ satisfies coordinate array $X^{v,v'}_{*}$, which is a contradiction by Lemma~\ref{lem:yes-1}.

\textbf{Case 3. $v\neq v'$ and nodes $v$ and $v'$ are not adjacent in configuration $H_{s'}$.}
In any configuration, the root node is adjacent to all other vertices, so $v$ and $v'$ must both be non-root nodes.
Suppose that in configuration $H_{s'}$, the root node is $w=\rho(H_{s'})$.
Since only leaf nodes in a configuration can be deleted, and since nodes $v$ and $v'$ are not deleted in $H_{s'},\dots,H_{s}$, we have that node $w$ exists and has degree at least two in each of $H_{s'},\dots,H_s$, and therefore must be the root node in each of $H_{s'},\dots,H_s$.
In particular, since the total order $\prec_H$ and root node of a configuration $H$ can only be changed when there are at most two vertices, no full operations from $H_{s'}$ to $H_s$ include flip operations.
Consequently, nodes $v$ and $v'$ have the same order with respect to orderings $\prec_{H_{s'}}$ and $\prec_{H_s}$

Assume without loss of generality that $v\prec_{H_{s'}} v'$ and $v\prec_{H_s} v'$ (the reverse direction is symmetric).
Let $t'$ be the largest index such that node $w$ is in configuration $H_{t'}$ ($t'\le 2k-3$ because configuration $H_{2k-2}$ only contains node $v'$). 
By maximality of $t'$, from configuration $H_{t'}$ to $H_{t'+1}$, node $w$ is deleted, so by Lemma~\ref{lem:yes-3}, $t'-s'\ge k-2$.
By Claim~\ref{lem:yes-1.6}, both $v$ and $v'$ are in $H_{s'-0.5}$.
Let $i_v$ be such that $v$ is the $i_v$th smallest node in configuration $H_{s'-0.5}$ according to $\prec_{H_{s'-0.5}}$.
Because $v\prec_{H_{s'-0.5}}v'$, and since configuration $H_{s'-0.5}$ is valid, Definition~\ref{def:edge} gives that stack $S_v(H_{s'-0.5})$ satisfies coordinate array $X^{v',w}_{v',i_v}$.
By Claim~\ref{lem:yes-2}, $(a_1,\dots,a_{k-s'})$ is a substack of $S_v(H_{s'-0.5})$, so by Lemma~\ref{lem:stack-1}, stack $(a_1,\dots,a_{k-s'})$ also satisfies coordinate array $X^{v',w}_{v',i_v}$.
On the other hand, by Claim~\ref{lem:yes-2}, $(a_k,\dots,a_{2k-1-t'})$ is a substack of $S_v(H_{t'+0.5})$.
Additionally, by Claim~\ref{lem:yes-1.6}, edge $(v',w)$ is also in $H_{t'+0.5}$, so stack $S_v(H_{t'+0.5})$, and thus stack $(a_k,\dots,a_{2k-1-t'})$, satisfies coordinate array $X^{v',w}_{v',i_v}$.
Since $k-s'\ge (2k-1-t)-1$, we have that for $j=k-s'$, stacks $(a_1,\dots,a_j)$ and $(a_k,\dots,a_{j+1})$ satisfy the same coordinate array $X^{v',w}_{v',i_v}$, which is a contradiction by Lemma~\ref{lem:yes-1}.

In all cases of $v$ and $v'$, we have shown a contradiction.
Thus, the path from configuration $H$ to configuration $H'$ in the Diameter instance $G$ cannot have length $2k-2$. 
Thus, when $A$ has $k$ orthogonal vectors, the Diameter of $G$ is at least $2k-1$.
This completes the proof.

%% file: k5.tex
\section{Main theorem for $k=5$}
\label{app:k5}
In this section, we prove Theorem~\ref{thm:main} (again) for $k=5$.
This proof shows how the $k=4$ proof in Section~\ref{sec:k4} can be easily modified to give a hardness reduction for $k=5$.
We include this proof because it is simpler than the $k=5$ instantiation of the general-$k$ proof in Section~\ref{sec:all}, so it may help to reader gain intuition for the general construction.
To avoid confusion, we highlight the main differences between the proof in this section and the general proof specialized to $k=5$.
\begin{itemize}
\item In the general proof specialized to $k=5$, vertices have up to three stacks. In this proof, vertices have up to two stacks. This difference is the main simplification.
\item To make this simplification work, we include ``coordinate change edges'' (as in the $k=4$ proof). By contrast, the general proof does not have such edges.
\item To make this simplification work, we also let coordinate arrays constrain stacks differently. In the general construction, if a coordinate array $x$ constrains two stacks $S$ and $S'$, that means both $S$ and $S'$ satisfy $x$. Here, we only require $S\circ S'$ or $S'\circ S$ to satisfy $x$.
\end{itemize}
\begin{theorem}
   Assuming SETH, for all $\varepsilon>0$ a $(\frac{9}{5}-\varepsilon)$-approximation of Diameter in \emph{unweighted, undirected} graphs on $n$ vertices needs $n^{5/4-o(1)}$ time.
   \label{thm:95}
\end{theorem}
\begin{proof}
Start with a 5-OV instance $\Phi$ given by a set $A\subset \{0,1\}^{\mathbb{d}}$ with $|A|=n_{OV}$ and $\mathbb{d} = c\log n_{OV}$.
We can check in time $n_{OV}^4$ where there are 4 orthogonal vectors in $A$, if so, we know $\Phi$ has 5 orthogonal vectors, so assume otherwise.
We construct a graph with $\tilde O(n_{OV}^4)$ vertices and edges from the 5-OV instance such that (1) if $\Phi$ has no solution, any two vertices are at distance 5, and (2) if $\Phi$ has a solution, then there exists two vertices at distance 9.
Any $(9/5-\varepsilon)$-approximation for Diameter distinguishes between graphs of diameter 5 and 9.
Since solving $\Phi$ needs $n_{OV}^{5-o(1)}$ time under SETH, a $9/5-\varepsilon$ approximation of diameter needs $n^{5/4-o(1)}$ time under SETH.
\paragraph{Construction of the graph}
The vertex set $L_1\cup L_2$ is defined on
\begin{align}
  L_1 &= \{(a,b,c,d)\in A^4\},  \nonumber\\
  L_2 &= \big\{(\{S_1,S_2\},x,y): \text{$S_1,S_2$ are stacks with $|S_1|+|S_2|=3$, }  \nonumber\\
   & \qquad\qquad\qquad\qquad\qquad \text{$x,y\in[\mathbb{d}]^3$ are coordinate arrays such that} \nonumber\\
   & \qquad\qquad\qquad\qquad\qquad \text{$S_1\circ S_2$ satisfies $x$ and $S_2\circ S_1$ satisfies $y$, OR} \nonumber\\
   & \qquad\qquad\qquad\qquad\qquad \text{$S_1\circ S_2$ satisfies $y$ and $S_2\circ S_1$ satisfies $x$}\big\}  
\end{align}
Throughout, we identify tuples $(a,b,c,d)$ and $(\{S_1,S_2\},x,y)$  with vertices of $G$, and we denote vertices in $L_1$ and $L_2$ by $(a,b,c)_{L_1}$ and $(\{S_1,S_2\},x,y)_{L_2}$ respectively.
The (undirected unweighted) edges are the following.
\begin{itemize}
\item ($L_1$ to $L_2$) Edge between $(a,b,c,d)_{L_1}$ and $(\{(a,b,c), ()\}, x,y)_{L_2}$ if stack $(a,b,c,d)$ satisfies both $x$ and $y$.
\item (vector change in $L_2$) For some vector $a\in A$ and stacks $S_1, S_2$ with $|S_1|\ge 1$, an edge between $(\{S_1,S_2\},x,y)_{L_2}$ and $(\{\popped(S_1),S_2+a\},x,y)_{L_2}$ if both vertices exist.
\item (vector change in $L_2$, part 2) For some vector $a\in A$ and stacks $S_1, S_2$ with $|S_1|\ge 1$, an edge between $(\{S_1,S_2\},x,y)_{L_2}$ and $(\{\popped(S_1)+a,S_2\},x,y)_{L_2}$ if both vertices exist.
\item (coordinate change in $L_2$) Edge between $(\{S_1,S_2\},x,y)_{L_2}$ and $(\{S_1,S_2\},x',y')_{L_2}$ if both vertices exist.
\end{itemize}

There are $n_{OV}^4$ vertices in $L_1$ and at most $n_{OV}^3\mathbb{d}^8$ vertices in $L_2$.
Note that each vertex of $L_1$ has $O(\mathbb{d}^8)$ neighbors, each vertex of $L_2$ has $O(n_{OV}+\mathbb{d})$ neighbors. 
The total number of edges and vertices, and thus the construction time, is $O(n_{OV}^4\mathbb{d}^8)=\tilde O(n_{OV}^4)$.
We now show that this construction has diameter 5 when $\Phi$ has no solution and diameter at least 9 when $\Phi$ has a solution.

\paragraph{5-OV no solution}
Assume that the 5-OV instance $A\subset\{0,1\}^{\mathbb{d}}$ has no solution, so that no five (or four or three or two) vectors are orthogonal.
We begin with the following lemma:
\begin{lemma}
  If stacks $(a,b)$ and $(a')$ satisfy $x$, then $(a,b,a')$ and $(a',a,b)$ satisfy $x$.
  If stacks $(a,b)$ and $(a',b')$ satisfy coordinate array $x$, then the stack $(a,b,b')$ satisfies coordinate array $x$.
  If stacks $(e',a',b')$ and $(a)$ satisfy coordinate array $x$, then stack $(a,a',b')$ satisfies coordinate array $x$.
\label{lem:5-no-3}
\end{lemma}
\begin{proof}
  For the first item, $(a,b,a')$ satisfies $x$ because $(a,b)$ satisfies $x$ and $a'$ is 1 in every coordinate of $x$.
  Similarly, $(a',a,b)$ satisfies $x$ because $a$ and $a'$ are 1 in every coordinate of $x$, and $b$ is 1 in at least 3 coordinates of $x$.

  For the second item, since $(a,b)$ and $(a',b')$ satisfy $x$, we have $a[x[i]]=1$ for $i\in [4]$, and there exists $I_2, J_2\subset[4]$ of size 3 such that $b[x[i]]=1$ for $i\in I_2$ and $b'[x[i]]=1$ for $i\in J_2$.
  We have $|I_2\cap J_2|=|I_2|+|J_2|-|I_2\cup J_2|\ge 3+3-4 = 2$.
  Thus, $I_1\supset I_2\supset (I_2\cap J_2)$ certifies that $(a,b,b')$ satisfies $x$.

  For the third item, because stack $(e',a',b')$ satisfies $x$, there exists $[4]=I_1\supset I_2\supset I_3$ with $a'[x[i]]=1$ for $i\in I_2$ and $b'[x[i]]=1$ for $i\in I_3$.
  Since $a[x[i]]=1$ for all $i\in [4]$, we thus have $I_1\supset I_2\supset I_3$ certifies that $(a,a',b')$ satisfies $x$.
\end{proof}

We show that any pair of vertices have distance at most 4, by casework on which of $L_1,L_2$ the two vertices are in.
\begin{itemize}
\item \textbf{Both vertices are in $L_1$:} Let the vertices be $(a,b,c,d)_{L_1}$ and $(a',b',c',d')_{L_1}$. 
By Lemma~\ref{lem:stack-2} there exists coordinate array $x$ satisfied by both stacks $(a,b,c,d)$ and $(a',b',c',d')$.
Then 
\begin{align}
  (a,b,c,d)_{L_1} 
  &- (\{(),(a,b,c)\},x,x)_{L_2}\nonumber\\
  &-(\{(a,b),(a')\},x,x)_{L_2}\nonumber\\
  &-(\{(a),(a',b')\},x,x)_{L_2}\nonumber\\
  &-(\{(),(a',b',c')\},x,x)_{L_2}
  -(a',b',c',d')_{L_1} 
\end{align}
is a valid path.
Indeed, the first edge and second vertex are valid because $(a,b,c,d)$ satisfies $x$ (and thus, by Lemma~\ref{lem:stack-1}, stack $(a,b,c)$ satisfies $x$).
By the same reasoning the last edge and fifth vertex are valid.
The third vertex is valid because $(a)$ and $(a',b')$ both satisfy $x$ and thus both $(a,a',b')$ and $(a',b',a)$ satisfy $x$ by the first part of Lemma~\ref{lem:5-no-3}.
By the same reasoning, the fourth vertex is valid.

\item \textbf{One vertex is in $L_1$ and the other vertex is in $L_2$ with stacks of size 1 and 2:} Let the vertices be $(a,b,c,d)_{L_1}$ and $(\{(a',b'),(e')\},x',y')_{L_2}$.
By Lemma~\ref{lem:stack-2}, there exists a coordinate array $x$ that is satisfied by stacks $(a,b,c,d)$ and $(a',b',e')$, and there exists a coordinate array $y$ satisfied by both stacks $(a,b,c,d)$ and $(e',a',b')$.
We claim the following is a valid path:
\begin{align}
  (a,b,c,d)_{L_1}
  &-(\{(a,b,c),()\},x,y)_{L_2} \nonumber\\
  &-(\{(a'),(a,b)\},x,y)_{L_2} \nonumber\\
  &-(\{(a',b'),(a)\},x,y)_{L_2} \nonumber\\
  &-(\{(a',b'),(e')\},x,y)_{L_2}
  -(\{(a',b'),(e')\},x',y')_{L_2}.
\end{align}

The first edge and second vertex are valid because $(a,b,c,d)$ satisfies $x$.

For the third vertex, we have $(a,b,c,d)$ and $(a',b',e')$ satisfy coordinate array $x$, so by Lemma~\ref{lem:stack-1}, stacks $(a,b)$ and $(a')$ satisfy coordinate array $x$.
Then by the first part of Lemma~\ref{lem:5-no-3}, stack $(a',a,b)$ satisfies $x$.
Similarly, $(a,b,c,d)$ and $(e',a',b')$ satisfy coordinate array $y$, so stacks $(a,b)$ and $(e',a')$ satisfy coordinate array $y$, so by the second part of Lemma~\ref{lem:5-no-3}, stack $(a,b,a')$ satisfies $y$.
Thus, the third vertex $(\{(a'),(a,b)\},x,y)_{L_2}$ is valid.

For the fourth vertex, we similarly have stacks $(a',b')$ and $(a)$ satisfy $x$, so  stack $(a',b',a)$ satisfy $y$.
Additionally, stacks $(e',a',b')$ and $(a)$ satisfy $y$ so $(a,a',b')$ satisfies $y$.
Thus the fourth vertex $(\{(a',b'),(a)\},x,y)_{L_2}$ is valid.

The fifth vertex  $(\{(a',b'),(e')\},x,y)_{L_2}$ is valid because $(a',b',e')$ satisfies $x$ and $(e',a',b')$ satisfy $y$ by construction of $x$ and $y$.

Hence, this is a valid path.

\item \textbf{Both vertices are in $L_2$ and have two stacks of size 1 and 2:} Let the vertices be $(\{(a,b),(e)\},x',y')_{L_2}$ and $(\{(a',b'),(e')\},x'',y'')_{L_2}$.
By Lemma~\ref{lem:stack-2}, there exists a coordinate array $x$ that is satisfied by $(a,b,e)$ and $(e',a',b')$, and there exists a coordinate array $y$ satisfied by both stacks $(e,a,b)$ and $(a',b',e')$.
Then the following is a valid path:
\begin{align}
  (\{(a,b),(e)\},x',y')_{L_2}
  &-(\{(a,b),(e)\},x,y)_{L_2} \nonumber\\
  &-(\{(a,b),(a')\},x,y)_{L_2} \nonumber\\
  &-(\{(a),(a',b')\},x,y)_{L_2}\nonumber\\
  &-(\{(e'),(a',b')\},x,y)_{L_2}
  -(\{(a',b'),(e')\},x'',y'')_{L_2}.
\end{align}
By construction of coordinate arrays $x$ and $y$, vertices $(\{(a,b),(e)\},x,y)_{L_2}$ and $(\{(a',b'),(e')\},x,y)_{L_2}$ are valid.
We now show vertex $(\{(a,b),(a')\},x,y)_{L_2}$ is valid, and the fact that vertex $(\{(a),(a',b')\},x,y)_{L_2}$ is valid follows by a symmetric argument.
We have stacks $(a,b)$ and $(e',a')$ satisfy $x$, so $(a,b,a')$ satisfies $x$ by the second part of Lemma~\ref{lem:5-no-3}.
Furthermore $(e,a,b)$ and $(a')$ satisfy $y$, so stack $(a',a,b)$ satisfies $y$ by the third part of Lemma~\ref{lem:5-no-3}.

\item \textbf{One vertex is in $L_2$ with two stacks of size 3 and 0:} 
For every vertex $u=(\{(a,b,c),()\},x,y)_{L_2}$ in $L_2$ with stacks of size 3 and 0, any vertex of the form $v=(a,b,c,d)_{L_1}$ in $L_1$ has the property that the neighborhood of $u$ is a superset of the neighborhood of $v$ (by consider coordinate change edges from $u$).
Thus, any vertex that $v$ can reach in 5 edges can also be reached by $u$ is 5 edges.
In particular, since any two vertices in $L_1$ are at distance at most 5, any vertex in $L_1$ is distance at most 5 from any vertex in $L_2$ with stacks of size 3 and 0.
Applying a similar reasoning, any two vertices in $L_2$ with stacks of size 3 and 0 are at distance at most 5, and any vertex in $L_2$ with stacks of size 3 and 0 is distance at most 5 from any vertex in $L_2$ with stacks of size 2 and 1.
\end{itemize}
We have thus shown that any two vertices are at distance at most 5, proving the diameter is at most 5.

\paragraph{5-OV has solution}
Now assume that the 5-OV instance has a solution.
That is, assume there exists $a_1,a_2,a_3,a_4,a_5\in A$ such that $a_1[i]\cdot a_2[i]\cdot a_3[i]\cdot a_4[i]\cdot a_5[i] = 0$ for all $i$.
Since we assume there are no 4 orthogonal vectors, we may assume that $a_1,a_2,a_3,a_4,a_5$ are pairwise distinct.

Suppose for contradiction there exists a path of length at most 8 from $u_0=(a_1,a_2,a_3,a_4)_{L_1}$ to $u_6=(a_5,a_4,a_3,a_2)_{L_1}$.
Since all vertices in $L_2$ have self-loops with trivial coordinate-change edges, we may assume the path has length exactly 8.
Let the path be $u_0=(a_1,a_2,a_3,a_4)_{L_1}, u_1,\dots,u_8=(a_5,a_4,a_3,a_2)_{L_1}$.
We may assume the path never visits $L_1$ except at the ends: if $u_i=(S)_{L_1}\in L_1$, then $u_{i-1}=(\{\popped(S),()\},x,y)_{L_2}$ and $u_{i+1}=(\{\popped(S),()\},x',y')_{L_2}$ are in $L_2$, and in particular $u_{i-1}$ and $u_{i+1}$ are adjacent by a coordinate change edge, so we can replace the path $u_{i-1}-u_i-u_{i+1}$ with $u_{i-1}-u_{i+1}-u_{i+1}$, where the last edge is a self-loop.

For $i=1,2,3,4$, let $p_i$ denote the largest index such that $u_0,u_1,\dots,u_{p_i}$ all contain a stack that has stack $(a_1,\dots,a_i)$ as a substack.
In this way, $p_4 = 0$.
For $i=1,\dots,4$, let $q_i$ be the smallest index such that vertices $u_{q_i},\dots,u_8$ all contain a stack with stack $(a_5,\dots,a_{6-i})$ as a substack.
In this way, $q_4 = 8$.
We show that, 
\begin{claim}
For $i=1,\dots,4$, between vertices $u_{p_i}$ and $u_{q_{5-i}}$, there must be a coordinate change edge.
\label{cl:5-yes}
\end{claim}
\begin{proof}
Suppose for contradiction there is no coordinate change edge between $u_{p_i}$ and $u_{q_{5-i}}$.

First, consider $i=4$.
Here, $u_{p_i}=u_0 = (a_1,a_2,a_3,a_4)_{L_1}$.
Then, $u_{q_1}$ is a vertex of the form $(\{S_1,S_2\},x,y)$ where $(a_5)$ is a substack of $S_1$.
Since there is no coordinate change edge, we must have $u_1 = (\{(a_1,a_2,a_3),()\},x,y)$ for the same coordinate arrays $x$ and $y$, so stack $(a_1,a_2,a_3,a_4)$ satisfies both $x$ and $y$.
Then $S_1$, and thus $(a_5)$, satisfies one of $x$ and $y$, so there is some coordinate array satisfied by both $(a_1,a_2,a_3,a_4)$ and $(a_5)$, which is a contradiction of Lemma~\ref{lem:yes-1}
By a similar argument, we obtain a contradiction with $i=1$.

Now suppose $i=3$.
Vertex $u_{p_3}$ is of the form $(\{(a_1,a_2,a_3),()\},x,y)$.
Then stack $(a_1,a_2,a_3)$ satisfies both coordinate arrays $x$ and $y$.
Vertex $u_{q_2}$ is of the form $(\{S_1',S_2'\},x,y)$ where $(a_5,a_4)$ is a substack of $S_1'$.
Then stack $S_1'\circ S_2'$ satisfies one of $x$ or $y$, and thus $(a_5,a_4)$, a substack of $S_1'\circ S_2'$, satisfies one of $x$ or $y$.
Thus, there is some coordinate array satisfied by both $(a_5,a_4)$ and $(a_1,a_2,a_3)$, which is a contradiction of Lemma~\ref{lem:yes-1}.
By a similar argument, we obtain a contradiction with $i=2$.

Thus, we have shown that for all $i=1,\dots,4$, there must be a coordinate change edge between $u_{p_i}$ and $u_{q_{5-i}}$.
\end{proof}

Since coordinate change edges do not change any vectors, by maximality of $p_i$, the edge $u_{p_i}u_{p_{i}+1}$ cannot be a coordinate change edge for all $i=1,\dots,4$.
Similarly, by minimality of $q_i$, the edge $u_{q_{i}-1}u_{q_i}$ cannot be a coordinate change edge for all $i=1,\dots,4$.

Consider the set of edges
\begin{align}
  \label{eq:5-yes-1}
  u_{p_4}u_{p_4+1}, u_{p_3}u_{p_3+1}, u_{p_2}u_{p_2+1}, u_{p_1}u_{p_1+1}, 
  u_{q_4-1}u_{q_4}, u_{q_3-1}u_{q_3}, u_{q_2-1}u_{q_2}, u_{q_1-1}u_{q_1}.
\end{align}
By above, none of these edges are coordinate change edges.
These edges are among the 8 edges $u_0u_1,\dots,u_7u_8$.
Additionally, the edges $u_{p_i}u_{p_i+1}$ for $i=1,\dots,4$ are pairwise distinct, and the edges $u_{q_i-1}u_{q_i}$ for $i=1,\dots,4$ are pairwise distinct.
Edge $u_{p_4}u_{p_4-1}$ cannot be any of $u_{q_i-1}u_{q_i}$ for $i=1,\dots,4$, because we assume our orthogonal vectors $a_1,a_2,a_3,a_4,a_5$ are pairwise distinct and $u_{p_4-1}=u_1$ does not have any stack containing vector $a_5$.
Similarly, $u_{q_4-1}u_{q_4}$ cannot be any of $u_{p_i}u_{p_i+1}$ for $i=1,\dots,4$.
Thus, the edges in \eqref{eq:5-yes-1} have at least 5 distinct edges, so our path has at most 3 coordinate change edges.
By Claim~\ref{cl:5-yes}, there must be at least one coordinate change edge.
We now casework on the number of coordinate change edges.

\textbf{Case 1: the path $u_0,\dots,u_8$ has one coordinate change edge.}
By Claim~\ref{cl:5-yes}, since vertex $u_{p_4}=u_0$ is before the coordinate change edge, edge $u_{q_1-1}u_{q_1}$ must be after the coordinate change edge, and similarly edge $u_{p_1}u_{p_1+1}$ must be before the coordinate change edge.
Then all of the edges in \eqref{eq:5-yes-1} are pairwise distinct, so then the path has 8 edges from \eqref{eq:5-yes-1} plus a coordinate change edge, for a total of 9 edges, a contradiction.

\textbf{Case 2: the path has two coordinate change edges.}
Again, by Claim~\ref{cl:5-yes}, for $i=1,\dots,4$, edges $u_{q_i-1}u_{q_i}$ must be after the first coordinate change edge, and edge $u_{p_i}u_{p_i+1}$ must be before the second coordinate change edge. 
Since we have 8 edges total, we have at most 6 distinct edges from \eqref{eq:5-yes-1}, so there must be at least two pairs $(i,j)$ such that the edges $u_{p_i}u_{p_i+1}$ and $u_{q_j-1}u_{q_j}$ are equal, and by above this edge must be between the two coordinate change edges.
Thus, each of $u_{p_4}u_{p_4+1}, u_{p_3}u_{p_3+1}, u_{p_2}u_{p_2+1}, u_{p_1}u_{p_1+1}$ and $u_{q_4-1}u_{q_4}, u_{q_3-1}u_{q_3}, u_{q_2-1}u_{q_2}, u_{q_1-1}u_{q_1}$ have at least two edges between the two coordinate change edges.
This means that vertices $u_{p_2},u_{p_1}, u_{q_2},u_{q_1}$ are all between the two coordinate change edges.
By Claim~\ref{cl:5-yes}, vertices $u_{p_3}$ and $u_{q_3}$ cannot be between the two coordinate change edges.
Thus, we must have $u_{p_1}u_{p_1+1}=u_{q_2-1}u_{q_2}$ and $u_{p_2}u_{p_2+1}=u_{q_1-1}u_{q_1}$.
Since we use at most 8 edges total and exactly 6 distinct edges from \eqref{eq:5-yes-1}, we have $q_1=p_1=p_2+1=q_2-1$.
However, this is impossible, because that means node $u_{p_1}=u_{q_1}$ has two stacks, one containing vector $a_1$ and one containing vector $a_5$. By maximality of $p_1$, the stack containing vector $a_1$ has no other vectors, and by minimality of $q_1$, the stack containing vector $a_5$ has no other vectors, so vertex $u_{p_1}=u_{q_1}$ has two stacks with a total of only two vectors, a contradiction of the definition of a vertex in $L_2$.

\textbf{Case 3: the path has three coordinate change edges.}
Since the distinct edges of \eqref{eq:5-yes-1} are 
\begin{align}
  \label{eq:5-yes-2}
  u_{p_4}u_{p_4+1}, u_{p_3}u_{p_3+1}, u_{p_2}u_{p_2+1}, u_{p_1}u_{p_1+1}, u_{q_4-1}u_{q_4},
\end{align}
we must have 
\begin{align}
  u_{p_3}u_{p_3+1} \ &= \   u_{q_1-1}u_{q_1} \nonumber\\
  u_{p_2}u_{p_2+1} \ &= \   u_{q_2-1}u_{q_2} \nonumber\\
  u_{p_1}u_{p_1+1} \ &= \   u_{q_3-1}u_{q_3} \nonumber\\
\end{align}
Hence, by Claim~\ref{cl:5-yes}, there must be a coordinate change edge between any two edges in \eqref{eq:5-yes-2}, so we must have four coordinate change edges, a contradiction.

This proves that there cannot be a length 8 path from $(a_1,a_2,a_3,a_4)$ to $(a_5,a_4,a_3,a_2)$, showing that the diameter is at least 9, as desired.

\end{proof}